\documentclass[reqno,12pt]{amsart}

\usepackage{amsmath,amssymb,amsthm,amsfonts}

\usepackage[widespace]{fourier}
\usepackage{times}

\usepackage{a4wide}
\usepackage{upref}
\usepackage[usenames,dvipsnames]{color}
\usepackage[bookmarks,colorlinks,breaklinks]{hyperref}  

\definecolor{dullmagenta}{rgb}{0.4,0,0.4}   
\definecolor{darkblue}{rgb}{0,0,0.4}
\hypersetup{linkcolor=red,citecolor=magenta,filecolor=dullmagenta,urlcolor=darkblue} 

\numberwithin{equation}{section}

\newtheoremstyle{ttheorem}%
       {1.3ex\@plus1ex}                
       {2.5ex\@plus1ex\@minus.5ex}      
       {\itshape}           
       {0pt}                   
       {\bfseries}          
       {.}                  
       {.5em}               
       {}                

\newtheoremstyle{ddefinition}%
       {1.3ex\@plus1ex}                
       {2.5ex\@plus1ex\@minus.5ex}      
       {}           
       {0pt}                   
       {\bfseries}           
       {.}                  
       {.5em}               
       {}                

\newtheoremstyle{rremark}%
       {1.3ex\@plus1ex}                
       {2.5ex\@plus1ex\@minus.5ex}      
       {\upshape}        
       {0pt}                   
       {\bfseries}           
       {.}                  
       {.5em}               
       {}                   

\theoremstyle{ttheorem}
\newtheorem{thm}{Theorem}[section]
\newtheorem{lem}[thm]{Lemma}

\newtheorem{cor}[thm]{Corollary}

\theoremstyle{ddefinition}
\newtheorem{defn}[thm]{Definition}

\theoremstyle{rremark}
\newtheorem{rem}[thm]{Remark}
\newtheorem{myremarks}[thm]{Remarks}
\newtheorem*{desMSA}{Description of the MSA}

\newenvironment{remarks}{\begin{myremarks}\begin{nummer}}%
    {\end{nummer}\end{myremarks}}

\newcounter{numcount}
\newcommand{\labelnummer}{\mbox{\normalfont (\roman{numcount})}}%

\makeatletter

\newenvironment{nummer}%
  {\let\curlabelspeicher\@currentlabel%
    \begin{list}{\labelnummer}%
      {\usecounter{numcount}\leftmargin0pt%
        \topsep0.5ex\partopsep2ex\parsep0pt\itemsep0ex\@plus1\p@%
        \labelwidth2.5em\itemindent3.5em\labelsep1em%
      }%
    \let\saveitem\item%
    \def\item{\saveitem%
      \def\@currentlabel{\curlabelspeicher$\,$\labelnummer}}%
    \let\savelabel\label%
    \def\label##1{\savelabel{##1}%
      \@bsphack%
        \ifmmode\else%
          \protected@write\@auxout{}%
          {\string\newlabel{##1item}{{\labelnummer}{\thepage}}}%
        \fi%
      \@esphack%
    }%
  }{\end{list}}%

\def\itemref#1{\expandafter\@setref\csname r@#1item\endcsname%
  \@firstoftwo{#1}}%

\def\section{\@startsection{section}{1}%
  \z@{1.3\linespacing\@plus\linespacing}{.5\linespacing}%
  {\normalfont\scshape\centering}}

\newcommand\cL{\mathcal{L}}
\newcommand\cU{\mathcal{U}}
\newcommand\cI{\mathcal{I}}
\newcommand\cJ{\mathcal{J}}
\newcommand\cT{\mathcal{T}}
\newcommand\cM{\mathcal{M}}
\newcommand\cF{\mathcal{F}}
\newcommand\cH{\mathcal{H}}
\newcommand\cS{\mathcal{S}}
\newcommand\cC{\mathcal{C}}
\newcommand\RR{\mathbb{R}}
\newcommand\NN{\mathbb{N}}
\newcommand\ZZ{\mathbb{Z}}
\newcommand\PP{\mathbb{P}}
\newcommand\DD{\mathbb{D}}
\newcommand\CC{\mathbb{C}}

\newcommand\eps{\varepsilon}

\renewcommand\P{\mathbb P}
\newcommand\E{\mathbb E}

\def\Chi{\raisebox{0ex}{$\chi$}}

\DeclareMathOperator{\supp}{supp}
\DeclareMathOperator{\dist}{dist}

\def\le{\leqslant}
\def\ge{\geqslant}
\let\leq\le
\let\geq\ge
\let\emptyset\varnothing

\newcommand{\be}{\begin{equation}}
\newcommand{\ee}{\end{equation}}
\newcommand{\benon}{\begin{equation*}}
\newcommand{\eenon}{\end{equation*}}
\newcommand{\ba}{\begin{array}}
\newcommand{\ea}{\end{array}}
\newcommand{\bal}{\begin{align}}
\newcommand{\eal}{\end{align}}
\newcommand{\bea}{\begin{eqnarray}}
\newcommand{\eea}{\end{eqnarray}}
\newcommand{\bee}{\begin{eqnarray*}}
\newcommand{\eee}{\end{eqnarray*}}

\newcommand{\Rd}{\mathbb R^d}

\newcommand{\norm}[1]{\Vert #1 \Vert}
\newcommand{\abs}[1]{\left| #1 \right|}
\newcommand{\Lp}[1]{\mathrm{L}^2(#1)}
\newcommand{\angles}[1]{\langle #1 \rangle}
\newcommand{\om}{\omega}
\newcommand{\pa}[1]{\left( {#1} \right)}

\renewcommand{\L}{\Lambda}

\begin{document}

\title[Localisation for Delone operators]{Localisation for Delone operators via Bernoulli randomisation}

\author[P.\ M\"uller]{Peter M\"uller}
\address[P.\ M\"uller]{Mathematisches Institut, Ludwig-Maximilians-Universit\"at,
	The\-re\-sien\-str.\ 39, 80333 M\"unchen, Germany}
\email{mueller@lmu.de}

\author[C. Rojas-Molina]{Constanza Rojas-Molina}
\address[C.\ Rojas-Molina]{Laboratoire AGM, Dpt. math\'ematiques, CY Cergy Paris Universit\'e,
2 av. Adolphe Chauvin, 95302 Cergy-Pontoise, France}
\email{crojasmo@u-cergy.fr}

\thanks{Version to appear in \emph{J.\ Anal.\ Math.}}

\begin{abstract}
	Delone operators are Schr\"odinger operators in multi-dimen\-sional Euclidean space with a potential given by
	the sum of all translates of a given ``single-site potential'' centred at the points of a Delone set.
	In this paper, we use randomisation to study dynamical localisation for families of Delone operators.
	We do this by suitably adding more points to a Delone set and by introducing i.i.d.\ Bernoulli random
	variables as coupling constants at the additional points. The resulting non-ergodic continuum Anderson model
	with Bernoulli disorder is accessible to the latest version of the multiscale analysis.
	The novel ingredient here is the initial length-scale estimate whose proof is hampered due to
	the non-periodic background potential. It is obtained by the use of a quantitative unique continuation principle.
	As applications we obtain both probabilistic and topological statements about dynamical localisation.
	Among others, we show that Delone sets for which the associated Delone operators exhibit dynamical
	localisation at the bottom of the spectrum are dense in the space of Delone sets.
\end{abstract}

\maketitle
\thispagestyle{empty}

%
\section{Introduction}
%

Since the discovery of quasicrystals by Schechtman \cite{Sch84}, aperiodic media has considerably attracted
the interest of both experimental and theoretical research. Aperiodic media do not exhibit translation invariance,
as crystals, but may possess structural long-range order, e.g.\ in the form of rotational symmetries.
They can be seen as intermediate structures between completely ordered ones (crystals) and disordered ones
(amorphous materials). They are, for this reason, expected to exhibit a wide range of electronic transport and
spectral properties. So far, the study of electronic properties of aperiodic media has been largely concentrated
on anomalous transport and singular continuous spectrum of the associated Hamiltonians for one-dimensional models
\cite{Bel03,BIST89,FuKr88,KoSu86,Su89}. In the present article, we aim to complement this description by studying
the absence of transport and the existence of bound states for Schr\"odinger operators associated to
aperiodic media in any dimension.

In order to describe the model in detail we introduce some notation. Let $d\geq1$ be the space dimension and
$\Lambda_{L}(x) := \raisebox{-.3ex}{\LARGE$\times$}_{j=1}^{d} (x_{j} - L/2, x_{j}+ L/2)$ be the open cube in
$\Rd$ with edges of length $L>0$ centred at $x=(x_{1},\ldots,x_{d})\in\RR^{d}$.
If the cube is centred about the origin, we simply write $\Lambda_{L} := \Lambda_{L}(0)$.

\begin{defn}
	\label{defdelone}
	A subset $D$ of $\Rd$ is called an $(r,R)$-\emph{Delone set}
	if ~(i)~~ it is \emph{uniformly discrete}, i.e.\ there exist a real $r>0$ such that
	$\big|D\cap\L_r(x)\big|\leq 1$  for every $x\in\Rd$, and ~(ii)~~ it is \emph{relatively dense},
	i.e.\ there exists a real $R\geq r$ such that $\big|D\cap\L_R(x)\big|\geq 1$ for every $x\in\Rd$.
	Here, $|\cdot|$ stands for the cardinality of a set.
\end{defn}

Clearly, the minimal distance between any two points in an $(r,R)$-Delone set is $r$.
Also, given any point in an $(r,R)$-Delone set, one can find another point that is no further than
$\sqrt{d}R$ apart. Particular examples of Delone sets are the hypercubic lattice $\mathbb{Z}^{d}$,
the vertices of a Penrose tiling or the random point set obtained from removing every other point of
$\ZZ$ by a Bernoulli percolation process.

In order to study electronic transport in a general aperiodic background of atoms,
we consider a one-body Schr\"odinger operator with
an atomic potential whose centres are given by the points of a Delone set.

\begin{defn}\label{defdeloneop}
	Given an $(r,R)$-Delone set $D$, we define the \emph{Delone operator} $H_D$ acting in $\Lp{\RR^d}$ by
	\be
		\label{delop}
		H_D :=-\Delta+ V_{D},
	\ee
	where $-\Delta$ is the $d$-dimensional Laplacian,
	\begin{equation}
		\label{delop-pot}
 		V_{D} := \sum_{\gamma\in D}u(\cdot-\gamma)
	\end{equation}
	and $u \in \mathrm{L}^{\infty}_{c}(\RR^{d})$ is a single-site potential -- a real-valued, bounded,
	measurable function with compact support.
	As the potential $V_{D}$ is bounded, $H_{D}$ is a self-adjoint operator on the domain of
	self-adjointness of $-\Delta$.
\end{defn}

Next we introduce a suitable metrisable topology $\tau$ on the space $\DD$ of all Delone sets.
There are several possible constructions for $\tau$. We recall a characterization given in
\cite[Lemma 3.1]{LS06}, see also \cite[Lemma~2.8]{MR}, that fits our purposes.
To this end, we also recall the Hausdorff distance $d_{H}(X,Y) := \max\big\{\sup_{x\in X}\inf_{y\in Y}|x-y|,
\sup_{y\in Y}\inf_{x\in X}|x-y|\big\}$ of two non-empty subsets $X,Y\subseteq \RR^{d}$, where
$|\boldsymbol\cdot|$ denotes the Euclidean norm on $\RR^{d}$.
\begin{defn}
 	\label{d:conv}
 	A sequence $(D_n)_{n\in\NN} \subset \DD$ of Delone sets converges to $D\in\DD$ in the topology $\tau$  if and only if, given any $l>0$, there exists $L>l$ such that the discrete sets $D_n\cap \L_L(0)$ converge to $D\cap \L_L(0)$ with respect to the Hausdorff distance
	as $n$ tends to infinity.
 \end{defn}
The map $\DD \ni D \mapsto H_D$, or restrictions of it, have been studied in the literature,
mostly from the point of view of ergodic dynamical systems.
Here, ergodicity is to be understood with respect to the translation group $\RR^{d}$ which acts on the hull or Delone dynamical system generated by $D$. It is defined as the closure with respect to $\tau$ of the set of all translations of $D$.
 For more details in the present context of Delone operators, we refer to, e.g.,  \cite{KLS03a, KLS03b,KLS11,LS02,LS03,LS05, LV09}. These works study
ergodic properties, like existence of a
deterministic spectrum and of the integrated density of states, both in our continuum
setting and in related discrete models.

Investigations on the spectral type or on dynamical properties of $H_{D}$ are much more subtle for
multi-dimensional systems. We only know two results of this kind (besides the periodic case):

\begin{nummer}
\item
  The first one is the ``soft'' approach by Lenz and Stollmann \cite{LS06} who elaborate on
  wonderland-type arguments based on Simon's pioneering work on singular continuous spectrum for
  Schr\"odinger operators \cite{Si}. They consider the subspace $\DD_{r,R}$ of all Delone sets with fixed
  parameters $r,R$ and show \cite[Thm.\ 3.2]{LS06} that under certain conditions on $r,R$ and $u$ there exists
  a dense $G_\delta$-set $\Omega_{sc} \subset \DD_{r,R}$ and a spectral subset $U\subset \RR$ such that for
  every $D\in\Omega_{sc}$, the spectrum of $H_D$ contains $U$ and is purely singular continuous there.
  As a consequence of this truly remarkable result, the existence of absolutely continuous spectrum or of
  pure point spectrum in $U$ is non-generic, that is, occurs at most for $D$ in a meagre subset of $\DD_{r,R}$.
\item
	The second result originates from the theory of random Schr\"odinger operators and was proven by Klopp,
	Loss, Nakamura and Stolz \cite{KLNS10}: given $\rho \in (0,1/2)$, they consider the subspace
	$\DD_{\rho}:= \{ \gamma + \omega_{\gamma}: \omega_{\gamma} \in \Lambda_{1-\rho}(0), \gamma\in \ZZ^{d}\}$
	of Delone sets which are perturbations of the hypercubic lattice. Viewing $\omega_{\gamma}$,
	$\gamma\in\ZZ^{d}$, as independent identically distributed (i.i.d.) random variables, each of which is, e.g.,
	uniformly distributed in  $\Lambda_{1-\rho}(0)$, the subspace $\DD_{\rho}$ becomes a probability space
	and gives rise to the name random displacement model, see also Remark~\ref{random-displ-m}.
	Apart from further regularity conditions, the single-site potential $u\ge 0$ is assumed to be supported in
	$\L_{\rho}(0)$. Then, \cite{KLNS10} show dynamical localisation of $H_{D}$ near the infimum of its spectrum
	with probability one for $D\in\DD_{\rho}$. In particular, this implies absence of propagation, as well as pure
	point spectrum with exponentially decaying eigenfunctions.
\end{nummer}

The results (i) and (ii) are of a very different nature -- in many ways. Whereas (i) states topological genericity
in $\DD_{r,R}$ of singular continuous spectrum, and thus topological non-genericity of localisation in a certain spectral interval, (ii) establishes probabilistic typicality of localisation at the bottom of the spectrum for lattice distortions. That these two seemingly conflicting properties can be two sides of the same coin is well known for the case of the discrete Anderson model  \cite{Si}.

In this paper, we investigate the occurrence of localisation for Delone operators $H_{D}$ both in a probabilistic sense and in a topological sense, and
without being restricted to lattice distortions for $D$. Our approach consists of a Bernoulli randomisation,
that is, we add new points to $D$ and introduce i.i.d.\ $0$-$1$-Bernoulli random variables $\omega$
as coupling constants at all new sites. This gives rise to a Delone operator denoted by $H_{D''^\om}$ below.
Here, the Delone set $D''^{\omega}$ consists of all points of $D$ and of all added points where the
Bernoulli coupling constants take the value one.
The operator $H_{D''^\om}$ is a random Schr\"odinger operator, more precisely, a continuum Anderson model with Bernoulli random variables in a non-ergodic setting. Random Schr\"odinger operators have been exhaustively studied in the literature \cite{CaLa90, PF92, S, K, AW}.
In the last decades, this field has seen many advances and the development of several methods to rigorously prove dynamical localisation. We require the most powerful and refined methods in our study of $H_{D''^\om}$ because of two difficulties: first, the non-ergodic setting due to the general nature (i.e.\ non-periodicity) of
both $D$ and of the set of added points; and, second, the singular (i.e.\ Bernoulli) nature of the random variables.

Random Schr\"odinger operators in a non-ergodic setting have been studied in the literature before in \cite{BdMNSS06,G,BdMLS11,RM12, RM13,GMRM, EK, ES, KK}. They are known to exhibit localisation at the bottom of the spectrum. However, the proofs rely heavily on the regularity of the random variables, and therefore do not apply to $H_{D''^\om}$. On the other hand,
Bernoulli random variables have been particularly difficult to treat with even within the most simple periodic setting. The first results on localisation in this case were restricted to one dimension, see
\cite{CKM,KLS90} for the lattice case and \cite{DSS02} for the continuum. In both cases, the proofs rely on the periodicity of the underlying structure of the random potential. It was only in 2005 when the work of Bourgain and Kenig \cite{BK} achieved a breakthrough and established localisation for a multi-dimensional continuum Anderson model with Bernoulli random variables. Their multiscale analysis was perfected in \cite{GKuniv} to yield the full strength of dynamical localisation and other related properties. We emphasise that the main step
in \cite{GKuniv} to start this multiscale analysis in the Bernoulli setting still relied on the periodicity of the model: Prop.\ 4.3 in \cite{GKuniv} requires Lifshitz tails, and Remark 4.4 in \cite{GKuniv} does not allow for a background potential. 

In Section~\ref{s:dl} we show how to perform the multiscale analysis from \cite{GKuniv} adapted to our
non-ergodic setting. This leads to Theorem~\ref{t:dynloc} on almost-sure dynamical
localisation for the non-ergodic continuum Anderson model with Bernoulli random variables  $H_{D''^\om}$
near the bottom of the spectrum.
Lemma~\ref{t:ilse} contains the novel idea to start this multiscale analysis by using a quantitative
unique continuation principle obtained in \cite{RMV12} and \cite{KV02}.
It does not rely on periodicity and works for Delone operators as considered here.

In Section~\ref{s:delops} we apply Theorem~\ref{t:dynloc} to obtain further information on the map
$D \mapsto H_{D}$. We show in Theorem~\ref{p:conv} that the set of Delone sets $D$ for which $H_{D}$ exhibits
dynamical localisation at the bottom of the spectrum is dense in $\DD$. In the remaining part of
Section~\ref{s:delops} we are concerned with the topological non-genericity of dynamical localisation of
$H_{D''^\om}$ in the
space of random coupling constants $\omega$ at the added sites. We present a ``soft'' and a ``hard'' approach
to this question, both relying on Theorem~\ref{t:dynloc}.
The ``soft'' approach in  Theorem~\ref{t:meagre} applies Simon's wonderland theorem \cite{Si} in the version of
\cite{LS06}. The ``hard'' approach in Theorem~\ref{t:md} performs an explicit analysis of the event of localised
configurations in the multiscale analysis of \cite{GKuniv}. While the former approach gives a stronger
statement, the latter is valid under more general assumptions. We point out that, in contrast to
\cite[Thm.\ 3.2]{LS06} mentioned in (i) above, our Theorems~\ref{t:meagre} and~\ref{t:md}
refer to the topology in the space of random coupling constants and allow for a specification of the relevant
spectral interval, denoted by $U$ in (i), to be located above the bottom of the spectrum.

In Appendix~\ref{specinf} we ensure that we work in the almost-sure spectrum of the random operator
$\omega \mapsto H_{D''^\om}$ despite the non-ergodic setting of the model.
In Appendix~\ref{A:QUCP} we track the constant in the unique continuation principle for one spatial dimension as given
in \cite{KV02}.
Finally, in Appendix~\ref{A:meagMSA} we gather auxiliary results on the multiscale analysis
that are needed for Section~\ref{s:delops}.

Parts of this work have been announced in \cite{RM20}. While finishing this paper, we learnt about the
preprint \cite{SeelmanTaufer19}, where the authors also use a quantitative unique continuation principle to
obtain an initial length-scale estimate. Their version even works near internal band edges,
see also Remark~\ref{rem:SeelT}.

%
\section{Bernoulli randomisation of Delone operators}\label{s:rso}
%

We will denote by $\Chi_\L$ the characteristic function of the set $\L$.

\begin{defn}\label{def:part}Given two Delone sets $D$ and $D'$, we say that $(D,D')$ is a \emph{Delone pair}, if $D'':=D\cup D'$ is again a Delone set.
\end{defn}

Given a Delone set $D$ with parameters $(r,R)$, there are many possibilities for a Delone set $D'$ with parameters $(r',R')$ such that
$D'':=D\cup D'$ is also a Delone set.
Here is an example: for every $\gamma \in D$ choose a point $\gamma'\in\Lambda_{\frac{r}{2}}(\gamma)\setminus\L_{\frac{r}{4}}(\gamma)$
and define $D':=\cup_{\gamma \in D}\{\gamma'\}$.
We obtain that $D''$ is a Delone set with parameters $(r/4,R)$.

\begin{defn}
	\label{def:delberop}
	A \emph{Delone--Bernoulli operator} associated to the Delone pair $(D,D')$ is the random operator
	\be
		\label{eq:defDelBern}
		\Omega_{D'} \ni \omega \mapsto H_{\omega} := H_{D''^\om}:=-\Delta+V_{D}+ V_{D'^\om},
	\ee
	acting in $\mathrm{L}^{2}(\RR^{d})$, where $\Omega_{D'}:=\{0,1\}^{D'}$, $-\Delta+V_D$ is the Delone
	operator associated to $D$ defined in \eqref{delop} and the random potential is given by
	\be
		\label{random-pot}
		V_{D'^\om}:= \sum_{\gamma\in D'}\omega_{\gamma} u(\cdot-\gamma)
	\ee
	with $\omega:=(\omega_{\gamma})_{\gamma\in D'}$.
	We assume that
	\begin{nummer}
	\item
		the single-site potential $u$ is the same in $V_{D}$ and $V_{D'^{\omega}}$ and satisfies
		\be
			\label{cond:u} u_-\Chi_{\L_{\delta_-}}\leq u \leq \Chi_{\L_{\delta_+}}
 		\ee
 		for some positive constants $\delta_\pm,u_- >0$ with $\delta_{-} <\delta_+<r'/4$.
	\item
		the random coupling constants  $\om_{\gamma}, \gamma\in D'$, are i.i.d.\ Bernoulli distributed,
		i.e., for some $\beta\in (0,1)$,
 		\be
			\mu(\om_{\gamma}=0)=\beta,\quad \mu(\om_{\gamma}=1)=1-\beta
		\ee
 		for every $\gamma \in D'$. The probability measure on $\Omega_{D'}$ is the product measure
		$\P_{D'} := \otimes_{\gamma\in D'}\mu$, and we denote the corresponding expectation by $\E_{D'}$.
	\end{nummer}
\end{defn}
With these assumptions, the map $\Omega_{D'}\ni\om\mapsto H_\omega$ is measurable and defines a random Schr\"odinger operator in the usual sense \cite{CaLa90,PF92}. Moreover, Assumption (i) implies that $\norm{V_D}_\infty,\sup_\omega\norm{V_{D'^\om}}_\infty\leq 1$, because the single-site potentials satisfy $\norm{u}_\infty<1$ and do not overlap.

\begin{rem}
The single-site potentials appearing in the definition of the potentials $V_D$ and $V_{D'^\om}$ can, in principle, be different. We model both potentials with the same single-site potential $u$ for simplicity only.
\end{rem}
Given $\omega\in \Omega_{D'}$, we define the point sets $D'^\omega :=\{ \gamma'\in D':\,\omega_{\gamma'}=1\}$,
which is not a Delone set $\P_{D'}$-a.s., and
\be D''^\om:=D\cup D'^\om \subset D''.\ee
The point set $D''^\om$ is a Delone set beause it contains the relatively dense set $D$ and is contained in the uniformly discrete set $D''$.
With this convention, we have $H_{\omega} = H_{D''^{\omega}} = -\Delta + V_{D''^{\omega}}$, and therefore the notations
\eqref{delop-pot} and \eqref{random-pot} are consistent.

We denote the spectrum of an operator $A$ by $\sigma(A)$. Let
\begin{equation}
	\label{infspec-H-D}
 	E_0:=\inf \sigma(H_D)
\end{equation}
be the bottom of the spectrum of the Delone operator $H_D := -\Delta+V_D$. If the set $D$ is infinitely pattern repeating  (see Definition~\ref{d:supf}) and if there exists $\epsilon >0$ such that $[E_{0}, E_{0} +\epsilon]\subset \sigma(H_{D})$, then $[E_{0}, E_{0} +\epsilon]\subset \sigma(H_\omega)$ for $\P_{D'}$-a.e. $\omega \in \Omega_{D'}$, see Theorem \ref{t:asspec}.

Given $\omega\in\Omega_{D'}$, $x\in\RR^{d}$ and $L>0$,
we denote by $H_{\omega,x,L}$ the restriction of $H_{\om}$ to the box $\L_L(x)$ with Dirichlet boundary conditions in the sense
\begin{equation}
 	H_{\omega,x,L} := H_{D,x,L} + V_{D'^\om \cap \Lambda_{L}(x)}\big|_{\Lambda_{L}(x)},
\end{equation}
where $H_{D,x,L}$ is the standard Dirichlet restriction of the background operator $H_{D}$ to $\Lambda_{L}(x)$.
For $z\in\CC\setminus \sigma(H_{\om})$, we denote the resolvent of $H_{\omega}$ by $R_{\om}(z):=(H_{\om}-z)^{-1}$ and, correspondingly, by $R_{\om,x,L}(z)$ that of $H_{\omega,x,L}$ for $z\in\CC\setminus \sigma(H_{\om,x,L})$. We denote by $\norm{\cdot}$ the $\rm L^2$-norm on either $\mathrm{L}^{2}(\RR^d)$ or
on $\mathrm{L}^{2}(\L_L(x))$, according to the context.

In the sequel, we drop the subscript $D'$ from $\PP_{D'}$ whenever the context is clear.
We denote by $C_{a,b,c,...}$ a constant depending on the parameters $a,b,c,...$

%
%
\section{Dynamical localisation at the bottom of the spectrum}\label{s:dl}

The most powerful multiscale analysis (MSA) available is the one of \cite{GKuniv}. It yields
dynamical localisation and all its consequences for continuum Anderson Hamiltonians with Bernoulli random variables.
In this section, we argue that it requires only minor modifications to make this MSA work for the more general
Delone--Bernoulli operators. The most involved step will be the verification of the initial length-scale estimate,
see Lemma~\ref{t:ilse} below, for which the method of \cite{GKuniv} does not work because it
relies crucially on periodicity. The main result of this section is summarised in Theorem~\ref{t:dynloc}.
We start by introducing the relevant notions.

\begin{defn}
 	We say that a random operator with self-adjoint realisations $H_{\omega}$ acting in $\mathrm{L}^{2}(\RR^{d})$ exhibits
	\emph{dynamical localisation} in an energy interval $I \subseteq \RR$, if there exist constants $M,a>0$ and
	$\vartheta \in (0,1)$ such that for a.e.\  $\omega$ there is a constant $C_{\omega}>0$ such that
	\begin{equation}
		\label{eq:loc-def}
 		\sup_{|f|\le 1} \norm{\Chi_{x}f(H_{\omega})\Chi_{I}(H_{\omega})\Chi_{y}}_{1} \le C_{\omega} \, e^{|x|^{a}} e^{-M |x-y|^{\vartheta}}
	\end{equation}
	for all $x,y\in \RR^{d}$.
	Here, $\norm{\cdot}_{1}$ denotes the trace norm, and the supremum is over all Borel measurable functions $f: \RR \rightarrow \CC$ which are pointwise bounded by one. The left-hand side of \eqref{eq:loc-def} is symmetric upon exchanging $x$ and $y$. Thus, one may replace
$|x|^{a}$ on its right-hand side by $\min\{|x|^{a},|y|^{a}\}$.

\end{defn}

Dynamical localisation in an interval $I$ is a strong version of Anderson localisation and implies the latter.
We use the notion Anderson localisation in $I$ for having dense pure point spectrum in $I$ and that
any eigenfunction corresponding to an eigenvalue in $I$ decays exponentially.

\begin{thm}\label{t:dynloc}
Let $\Omega_{D'} \ni \omega \mapsto H_{\omega}$ be the Delone--Bernoulli operator from Definition \ref{def:delberop}. Then,
\begin{nummer}
\item if $d\geq2$, for any $\beta\in(0,1)$, there exists an energy $E_*>E_0$ and an event $\Omega_{MSA}$ with $\PP(\Omega_{MSA})=1$ such that $H_{\om}$ exhibits dynamical localisation in $[E_0,E_*]$ for every $\omega\in\Omega_{MSA}$.
    \item if $d=1$, there exists a positive constant $C\leq 1$ such that if $\beta\in(0,C)$, the same result as in \textup{(i)}\ holds.
\end{nummer}
\end{thm}

\begin{rem} The restriction on the disorder parameter $\beta$ in $d=1$ is due to the use of quantitative unique continuation principles in the proof. While this restriction is not present in previous works on the Bernoulli--Anderson model, our method of proof allows to treat the case of a non-zero background Delone potential $V_D$, which is not possible with other existing methods, including the specific one-dimensional analysis in \cite{DSS02}.
\end{rem}

\begin{desMSA}
The usual proof of localisation via a multiscale analysis (see e.g. \cite{FS,CH1,vDK89,GK1}) is an induction procedure on the scale $L$ of the finite-volume operator $H_{\om,x,L}$. Generally, the goal is to show that the resolvent of the operator on finite volumes decays fast enough along a sequence of scales $L_k$ that grows to infinity. There are two main ingredients to achieve this:
\begin{nummer}
    \item To show that there is an initial scale $L_0$ for which the finite-volume resolvent decays at a good rate between distant points in space, with good probability.
    \item To show that, if the finite-volume resolvent decays at a good rate at a certain scale $L_k$ with good probability, then it also decays at a good rate in the next scale $L_{k+1}$ with good probability.
\end{nummer}
Step (i) is known as the \emph{initial length-scale estimate}, while (ii) is the multiscale procedure itself. What a \emph{good rate with good probability} means, is to be defined within the method and can vary. In our case, good rate of decay with good probability means that the finite-volume resolvent is exponentially decaying, with a probability algebraically close to one. This brings us to the definition of a good box (\cite[Sect.~3]{GKuniv}).
\end{desMSA}

\begin{defn}\label{def:gbox}
Let $\omega\in\Omega_{D'}$, $E\in\CC$ a (complex) energy, $m>0$ a rate of decay and $\zeta\in(0,1)$. A box $\L_L(x)$ is $(\om,E,m,\zeta)$-good if 
\be\label{gbox1} \norm{R_{\om,x,L}(E)}\leq e^{L^{1-\zeta}},
\ee
and
\be\label{gbox2} \norm{\Chi_y R_{\om,x,L}(E)\Chi_z}\leq e^{-m\norm{y-z}}
\ee
for all $y,z \in\L_L(x)$ with $\norm{y-z}\geq {L}/{100}$.
\end{defn}

\begin{defn}\label{def:gscale}
Let $E\in\CC$ be  a (complex) energy, $m>0$ a rate of decay, $\zeta\in(0,1)$ and $p>0$.
A scale  $L>0$ is $(E,m,\zeta,p)$-good if, uniformly with respect to $x\in\RR^d$, we have
\be\label{gscale} \P\Big( \L_L(x) \mbox{ is }\,(\om,E,m,\zeta)\mbox{-good} \Big)\geq 1-L^{-pd}.
\ee
\end{defn}
Then, in the terminology of \cite{GKuniv}, the statement that the finite-volume resolvent $R_{\om,x,L}$ decays at a good rate with a good probability, for any box of side length $L$, is equivalent to the statement that the scale $L$ is good for a certain choice of parameters in the sense of Definition \ref{def:gscale}.

We obtain \cite[Thm.\ 4.1]{GKuniv} in our setting:

\begin{thm}\label{t:msa}
Let $\Omega_{D'} \ni \omega \mapsto H_{\omega}$ be the Delone--Bernoulli operator from Definition \ref{def:delberop}. Fix $p\in]\frac{1}{3},\frac{3}{8}[$ and $\zeta\in(0,1)$. Assume either $d\geq 2$ and $\beta \in (0,1)$, or $d=1$ and $\beta \in (0,C)$ for a certain constant $C\in(0,1]$. Then there exists an energy $E_*>0$, a rate of decay $m$ and a scale $L_0$, all depending on the parameters of the model $D,D', u,\beta,p,\zeta$, such that all scales $L\geq L_0$ are $(E,m,\zeta,p)$-good for all energies $E\in[E_0,E_*]$.
\end{thm}

\begin{proof}[Proof of Theorem \ref{t:dynloc}]
Theorem \ref{t:msa} allows us to follow the construction of a set $\mathcal U\subset \Omega_{D'}$ for which the associated resolvents have all good properties, as stated in \cite[Thm.\ 6.1]{GKuniv}. This, in turn, allows to follow the proof of dynamical localisation in \cite[Sect.~7]{GKuniv}. One then obtains a generalisation of \cite[Thm.\ 7.2 and Cor.\ 7.3]{GKuniv}, of which Theorem \ref{t:dynloc} is the particular case for Bernoulli random variables.
\end{proof}

It remains to prove Theorem \ref{t:msa}, i.e., we need to show that we can perform the MSA of \cite{GKuniv}.
Note that the MSA in \cite{BK,GKuniv} is more involved than what we briefly described above: instead of dealing with one scale at each step of the iteration, the authors of \cite{BK,GKuniv} handle several scales at each step. This is needed to treat singular random variables. In the case of regular random variables, the Wegner estimate allows to control the fluctuations of the eigenvalues of the finite-volume operator (see e.g. \cite{K,S, AW}), that are induced by variations in the random parameters $\omega_\gamma$. The Wegner estimate is a crucial ingredient in each iteration step (ii) of the MSA described above.
In the case of Bernoulli random variables, however, this estimate is absent (except in $d=1$, see \cite{DSS02}).
Fluctuations in the parameters $\omega_\gamma\in\{0,1\}$ need not necessarily generate an according fluctuation in the eigenvalues of $H_{\om,x,L}$. Bourgain and Kenig tackled this problem in \cite{BK} by proving a weaker version of the Wegner estimate at each step of the multiscale iteration, which led to the need of dealing with several scales at each step. Their argument, and the following refinement in \cite{GKuniv}, was based on the use of so-called \emph{free sites}, and a quantitative version of the unique continuation principle. In their periodic setting, the free sites form a subset of lattice points $y\in \ZZ^d$, abundant enough, where the random parameters $\omega_y\in\{0,1\}$ are replaced by deterministic continuous parameters $t_y\in[0,1]$. These free sites are used to generate fluctuations of the eigenvalues, which can be estimated using quantitative versions of the unique continuation principle for the corresponding eigenfunctions.

\begin{defn}[An operator with free sites] \
\begin{nummer}
\item Let $D'$ be a Delone set with radius of relative denseness $R'$.
	For each (sub-) lattice point $z\in (2R'\ZZ)^{d}$ choose a point $\gamma_z\in D'\cap \L_{R'}(z)$.
	(Note that by relative denseness such a point always exists.). We define the $(R',3R')$-Delone set
\be D_0:=\{\gamma_z: z\in (2R'\ZZ)^{d}\} \subseteq D'\ee
and the set of \emph{free sites}
\be S:= D'\setminus D_0, \ee
which is a Delone set with radius of relative denseness $2R'$.
\item \label{fvfreesites}
	Given a Delone pair $(D,D')$, a set of free sites $S\subset  D'$, parameters $t_S:= (t_\gamma)_{\gamma\in S}\in[0,1]^S$ and a box $\L_L(x)$, we introduce the \emph{finite-volume Delone--Bernoulli operator with free sites}
    \be H_{\om,t_S,x,L} := H_{D,x,L} + V_{\om,t_S,x,L}\quad \mbox{in~} \,\Lp{\L_L(x)},
    \ee
    where, for $y\in\L_L(x)$,
    \be V_{\om,t_S,x,L}(y) := \sum_{\gamma\in D_{0} \cap \L_L(x)}\omega_\gamma u(y-\gamma) + \sum_{\gamma\in S\cap \L_L(x)}t_\gamma u(y-\gamma).
    \ee
    \end{nummer}
\end{defn}

\begin{proof}[Proof of Theorem \ref{t:msa}]
	This is the analogue of \cite[Thm.\ 4.1]{GKuniv}, which relies upon the initial length-scale 
	estimate \cite[Prop.\ 4.3]{GKuniv} and the induction procedure \cite[Prop.\ 4.6]{GKuniv}. 
	The statements of \cite[Thm.\ 4.1, Prop.\ 4.3, Prop.\ 4.6]{GKuniv} involve 
	operators with free sites and hold uniformly in the coupling constants $t_{S}$ at the free sites. 
	Uniformity in $t_{S}$ is crucial to perform the induction step to proceed to larger and larger length scales. 
	However, in the final result of \cite[Thm.\ 4.1]{GKuniv}, the coupling constants $t_{S}$ 
	can be particularly chosen as the realisations of the given Bernoulli random 
	variables at the free sites, which then amounts to the claim of Theorem~\ref{t:msa} for the (original) 
	Delone--Bernoulli operator. We proceed to discuss the two ingredients Prop.\ 4.3 and Prop.\ 4.6 in \cite{GKuniv}.
	
	Prop.\ 4.6 in \cite{GKuniv} is applicable to our setting: Firstly, our choice of free sites is abundant 
	in the sense of \cite[Def.\ 3.8(i)]{GKuniv}. Secondly, we follow a reasoning similar to \cite{RM12}: 
	As our Definition~\ref{def:gscale} of a good scale involves an estimate that is uniform with respect 
	to the centre of the box, periodicity of the potential does not play a role any more. 
	The statements remain valid for 
	the potential $V_D+V_{D'^\om}$ with an underlying Delone structure.
	Thirdly, we employ the quantitative version of the unique continuation principle \cite{RMV12} 
	that applies to our setting. This is used to prove the crucial result \cite[Lemma 4.14]{GKuniv}, 
	which is a weak version of the Wegner estimate. It is the proof of this lemma where the power of working with  
	operators with free sites from Definition \ref{fvfreesites} comes into effect. Altogether, this establishes 
	the validity of \cite[Prop.\ 4.6]{GKuniv} in our setting. 

	Concerning the initial length-scale estimate \cite[Prop.\ 4.3]{GKuniv}, the proof given in 
	\cite[Sect.\ 4.1]{GKuniv} relies heavily on the periodicity of the background potential. 
	As this is not necessarily the case for $V_D$, these proofs do not apply to our setting.
	Therefore, the rest of our analysis will focus on proving the initial length-scale estimate 
	uniformly in the free sites in our setting.
	This is done in the next Lemma~\ref{t:ilse}, which then completes the proof of Theorem~\ref{t:msa}.
\end{proof}

\begin{lem}\label{t:ilse}
Let $(D,D')$ be a Delone pair and $H_{\omega,t_{S},x,L}$ the finite-volume Delone--Bernoulli operator with free sites from Definition \ref{fvfreesites}. We fix $p>0$ and $\beta\in(0,1)$.
\begin{nummer}
\item
Let the dimension $d\geq 2$ and $\varepsilon\in(0, \frac34 - \frac1d)$. Then there exists $\tilde L=\tilde L(d,D,D',u,\beta,p,\varepsilon)>0$ and constants $C_{d,u,D',\varepsilon}, C_{1}, C_{2} >0$, where $C_{1},C_{2}$ depend on the parameters of the model, such that for all scales $L\geq \tilde L$ and all $x\in\RR^d$, we have
\be\label{ilse}
\P\left( H_{\om,t_S,x,L}\geq E_0+ \mathcal{E}_{L},\,\,\forall t_S\in[0,1]^S \right)\geq 1-L^{-pd},
\ee
where $\mathcal{E}_{L} := u_-(C_{d,u,D'}\log L)^{-(\varepsilon +1/d)[C_1+ C_2(\log L)^{(4/3)(\varepsilon +1/d)}]}$.
It follows that there exists $m_{L}>0$ and $\zeta \in (0,1)$ such that, for all energies $E\in[E_0,E_0 + \mathcal{E}_{L}/2]$,
all scales $L\geq \tilde L$ are $(E,m_L,\zeta,p)$-good for the operator with free sites, uniformly in $t_{S} \in [0,1]^{S}$.
\item In dimension $d=1$
there exists $\tilde L=\tilde L(D,D',u,\beta,p 
)>0$
such that for all scales $L\geq \tilde L$ and all $x\in\RR^d$, \eqref{ilse} holds with $d=1$ and $\mathcal{E}_{L}$ replaced by
$\widetilde{\mathcal{E}_{L}}:= C_{u, D,D',p} |\log\beta| \; L^{-2 \tilde C /|\log \beta|} / \log L$, where $\tilde C =
\tilde C_{u, D,D',p}>0$. Moreover, if $\beta \in (0, e^{-\tilde C})$,
then there exists $m_{L}>0$ and $\zeta \in (0,1)$ such that, for all energies $E\in[E_0,E_0 + \widetilde{\mathcal{E}}_{L}/2]$,
all scales $L\geq \tilde L$ are $(E,m_L,\zeta,p)$-good for the operator with free sites, uniformly in $t_{S} \in [0,1]^{S}$.
\end{nummer}
\end{lem}

\begin{rem}
	\label{rem:SeelT}
	The arguments for Lemma~\ref{t:ilse} are based on quantitative unique continuation principles and convenient
	configurations of the random potential. They appear in a preliminary form in the PhD thesis of the second
	author \cite{RMthesis}. While finishing this paper, we learnt that similar arguments are used in
	the recent preprint \cite[Thm.\ 2.2]{SeelmanTaufer19} to obtain an initial length-scale estimate
	which is not restricted to the bottom of the spectrum but also works at internal band edges.
	It is applied to yield localisation at band edges in a periodic situation
	\cite[Thm.~1.1]{SeelmanTaufer19}, but the authors also speculate about applications in certain non-ergodic
	ones \cite[Ex.~2.4, Ex.~2.5]{SeelmanTaufer19}.
	We could combine  \cite[Thm.\ 2.2]{SeelmanTaufer19} with the multi-scale analysis
	performed in this section to obtain a generalisation of Theorem \ref{t:dynloc} to energies near
	internal band edges -- provided they exist for $H_{D}$.
\end{rem}

{
Before we turn to the proof of Lemma~\ref{t:ilse}, we recall the Combes--Thomas estimate from
\cite[Thm.~1]{GKpams} in our setting ($n=1,z=E,\Xi=1$, $\gamma=1/2$ therein), which also holds for finite-volume restrictions.

\begin{thm}[\cite{GKpams}]\label{t:CT}
	Let $\omega\in \Omega_{D'}$, $L>0$, $x\in\RR^{d}$ and $y,z \in\L_{L}(x)$. Suppose
$E< \inf \sigma(H_{\om,x,L})$ and write $\eta:=\inf \sigma(H_{\om,x,L})-E >0$. Then,
\be \norm{\Chi_y R_{\om,x,L}(E)\Chi_z }\leq \frac{4}{3\eta}e^{\frac{\sqrt{\eta}}{2}\left(\sqrt{d}- \norm{y-z}\right)}. \ee
\end{thm}
}

\begin{proof}[Proof of Lemma \ref{t:ilse}]
\emph{Part} (i). \quad Let $d \ge 2$.
We fix $x\in \RR^{d}$ and the configuration $t_S\in[0,1]^S$. First, we prove \eqref{ilse}.
Since $u\geq0$, we have
\be
	\label{tS-lower}
	H_{\om,t_S,x,L}\geq H_{D,x,L} + V_{D_{0}^{\omega}\cap \L_{L}(x)}\big|_{\Lambda_{L}(x)},
\ee
where $D_0^{\omega} := \{ \gamma \in D_{0}: \omega_{\gamma}=1\}$. Recall that $D_{0}\subset D'$ is a
$(R',3R')$-Delone set. We write $R_{0} := 3 R'$ for its radius of relative denseness.
Therefore, in order to obtain the event in \eqref{ilse}, it is enough to prove a suitable
lower bound for the operator on the r.h.s.\ of \eqref{tS-lower}, the result being automatically uniform in $t_S$.

We consider $L>4R_{0}$ which allows to fix a real $M > R_{0} + 2 \delta_{+}$ such that $\ell := L/M \in 2\NN+1$
is an odd natural number. We define the centred discrete cube $\hat\Lambda_\ell := \Lambda_{\ell} \cap \ZZ^{d}$
consisting of $|\hat\Lambda_{\ell}| = \ell^{d}$ lattice points. This allows the decomposition of
\be
	\label{del}
	\L_{L}(x) = \bigg(\bigcup_{k \in \hat\Lambda_{\ell}} \overline{\L_M(x + Mk)} \bigg)^{\rm int}
 \ee
into $\ell^{d}$ boxes of side length $M$ with mutually disjoint interior.
We introduce the event
\be
	\mathcal A := \Big\{ \om\in \Omega_{D'}: \big|D_0^{\omega} \cap \L_{M-2\delta_+}(x+Mk)\big| \ge 1
	\text{~~for every } k\in\hat\Lambda_{\ell}\Big\}
\ee
of finding at least one point of $D_{0}^{\om}$ well inside each cube in the decomposition
\eqref{del}.
The complement of this event has probability
\begin{align}
	\P(\mathcal A^c) &=\P\left({\exists\, k\in\hat\Lambda_{\ell} ~ \forall\gamma\in  D_{0} \cap
		\L_{M-2\delta_+}(x+Mk)\,:\,\omega_{\gamma}=0 }\right)  \notag \\
 	& \leq \sum_{k\in\hat\Lambda_{\ell}} \P\pa{\forall \gamma\in D_{0} \cap \L_{M-2\delta_+}(x+Mk)\,\,: \,
		\omega_{\gamma}=0  } \notag \\
	& \leq |\hat\Lambda_{\ell}| \; [\P(\omega_{\gamma}=0)]^{\min_{k\in\hat\Lambda_{\ell}}| D_{0}
		\cap\L_{M-2\delta_+}(x+Mk)|} \notag \\
	& \leq \ell^d \beta^{\lfloor{(M-2\delta_+)/{R_0}}\rfloor^d},
\end{align}
where $\lfloor x\rfloor$ denotes the largest integer not exceeding $x\in\RR$. Therefore, we have
\be
	\label{A-prob}
	\P(\mathcal A)\geq 1 - \ell^d \beta^{\lfloor{(M-2\delta_+)/{R_0}}\rfloor^d}.
\ee
From now on we fix an arbitrary $\omega\in \mathcal A$. For every $k\in\hat\Lambda_{\ell}$, this allows
to choose one point in $D_0^{\omega}\cap\L_{M-2\delta_{+}}(x+Mk)$. We denote this point by
$x+ M\gamma_k(\omega)$, that is,
\be
	\label{eq:gamma-loc}
	\gamma_{k}(\omega) \in \Lambda_{1- 2\delta_{+}/M}(k), \quad  k\in\hat\Lambda_{\ell}.
\ee
Thus, we obtain the lower bound
\be
	\label{eq:thinned-out}
	V_{D_{0}^{\omega}\cap \L_{L}(x)}(y) \ge \sum_{k\in\hat\Lambda_{\ell}} u\big(y - x - M\gamma_k(\omega)\big)
	 =: W_{\omega}(y-x),
	\quad 	y\in \Lambda_{L}(x),
\ee
for the potential in \eqref{tS-lower}.
We conclude from \eqref{tS-lower}, \eqref{eq:thinned-out} and unitary invariance under translation by $x$ that
\be
	\label{eq:first-transformed}
	\inf \sigma(H_{\om,t_S,x,L}) \ge \inf \sigma (H_{D-x,0,L} + W_{\omega} ).
\ee
Here, $D-x:=\{ \gamma-x:\,\gamma\in D\}$ is the Delone set $D$ translated by $-x$.

We will apply the scale-free unique continuation principle from \cite{RMV12} to the operator on the
right-hand side of \eqref{eq:first-transformed} to estimate the minimal lifting of the ground-state energy of
$H_{D-x,0,L}$ by the potential $W_{\omega}$. In order to adapt our problem to the formulation in \cite{RMV12},
we first need another unitary transformation.
The unitary dilation operator
$S_{M}: {\mathrm L}^{2}(\Lambda_{L}) \rightarrow \mathrm L^{2}(\Lambda_{\ell})$,
$\psi \mapsto M^{d/2} \psi(M\,\boldsymbol\cdot)$, transforms the operator on the right-hand side
of \eqref{eq:first-transformed} according to
\be
	\label{eq:op-conjugate}
	S_{M}( H_{D-x,0,L} + W_{\omega}) S_{M}^{*}
	= M^{-2} \big( H_{D}^{(M)} + M^{2} W_{\omega}(M\,\boldsymbol\cdot)\big),
\ee
where $H_{D}^{(M)} := -\Delta_{\ell} +  M^{2} V_{D-x}(M\,\boldsymbol\cdot)$ is the transformed background
operator with Dirichlet boundary conditions on $\mathrm L^{2}(\Lambda_{\ell})$
and the transformed potential can be written as
\be
	\label{eq:W-alpha-def}
	W_{\omega}^{(M)} := M^{2} W_{\omega}(M\,\boldsymbol\cdot) = \sum_{k\in\hat\Lambda_{\ell}}
	u_{M}\big(\boldsymbol\cdot \,- \gamma_{k}(\omega)\big)
\ee
with
\be
	\label{eq:uM}
	u_{M} := M^{2} u(M\,\boldsymbol\cdot) \ge M^{2} u_{-} \Chi_{\Lambda_{\delta_{-}/M}}.
\ee
The decomposition \eqref{del} transforms into
\be
	\label{del-alpha}
	\L_{\ell} = \bigg(\bigcup_{k \in \hat\Lambda_{\ell}} \overline{\L_{1}(k)}
	\bigg)^{\rm int}.
 \ee
Each of the unit cubes in this decomposition feels the potential $W_{\omega}^{(M)}$ sufficiently
strongly in the sense that
\be
	\label{eq:strong-pot}
	\supp \Chi_{\Lambda_{\delta_{-}/M}} \big(\boldsymbol\cdot \,- \gamma_{k}(\omega)\big)
	\subset \Lambda_{1}(k)
\ee
for every $k\in\hat\Lambda_{\ell}$; see \eqref{eq:gamma-loc}.

We conclude from \eqref{eq:W-alpha-def}, \eqref{eq:uM} and \eqref{eq:strong-pot} that the assumptions
of the scale-free unique continuation principle
\cite[Thm.\ 4.9(a)]{RMV12} are satisfied for the operator $H_{D}^{(M)} + W_{\omega}^{(M)}$ with our $\ell$ playing the
r\^ole of $L$ there and with $\delta =\delta_{-}/M$, $C_{-} = M^{2} u_{-}$, $x=0$ and $t=1$ there.
Moreover, as the unperturbed operator is unitarily conjugate to the Dirichlet restriction
$S^{*}_{M}H_{D}^{(M)}S_{M} = M^{2}H_{D-x,0,L}$, we infer
\be
	\inf\sigma(H_{D}^{(M)}) = M^{2} \inf\sigma(H_{D,x,L}) \ge M^{2} E_{0},
\ee
and \cite[Thm.\ 4.9(a)]{RMV12} provides the existence of a constant $C_{UC}(M)>0$ such that
\be
	\label{lowb}
	\inf \sigma(H_{\om,t_{S},x,L}) = M^{-2} \inf \sigma(H_{D}^{(M)} + W_{\omega}^{(M)}) \geq  E_0 + u_- C_{UC}(M).
\ee
Here we used \eqref{eq:first-transformed} and \eqref{eq:op-conjugate} for the first equality.
The constant $C_{UC}(M)$ is determined by \cite[Eq.\ (7)]{RMV12} and involves a certain quantity $K_{V}^{2/3}$.
The relevant definition of $K_{V}$ in the context of \cite[Thm.\ 4.9(a)]{RMV12} is given in the comment before
\cite[Eq.\ (33)]{RMV12}. We estimate
\begin{align}
	\label{KV-bound}
	K_{V}
	& \le 2 \norm{M^{2} V_{D-x}(M\,\boldsymbol\cdot)}_\infty
		+ 2 \sup_{\omega'\in\mathcal{A}}\norm{W_{\omega'}^{(M)}}_\infty
		+ \sup_{\ell'\in 2\NN +1} \inf \sigma(-\Delta_{\ell'})  \notag\\
	& = M^{2} (4 + M^{-2} d \pi^{2}/3^{2})
\end{align}
and conclude
\be
	\label{c:ucp}
	C_{UC}(M):=\pa{\frac{\delta_-}{C_d M}}^{C_d + 3 C_d  M^{4/3}},
\ee
where $C_{d}$ depends only on the dimension. Note that the use of the upper bound $K_{V}<5M^2<3^{3/2}M^2$ in
\eqref{c:ucp} is justified because we may assume that $M$ is large enough and therefore $\delta_-/(C_dM) <1$.
This argument may require to enlarge $L$. But we will be choose $M$ as an increasing function of $L$
in \eqref{KLchoice} below.
We summarise that the spectral lifting in \eqref{lowb} relative to the spectral infimum $E_{0}$ of the
background operator is independent of $\omega\in\mathcal A$, $x$ and $L$.

Combining \eqref{lowb} and \eqref{A-prob}, we obtain
\be
	\label{est:ilseprob2}
	\P\Big( \inf \sigma(H_{\omega,t_{S},x,L}) \geq E_0+u_-C_{UC}(M)\Big)
	\geq 1 - \pa{{L}/{M}}^d\beta^{\lfloor{(M-2\delta_+)/{R_0}}\rfloor^d}.
\ee
Given $\eps>0$, which will be determined later, we make the choice
\be
	\label{KLchoice}
	M=M_L:= \max\Big\{ \mu \in \big(R_{0}+ 2\delta_{+}, 4R_0(\log L)^{\epsilon + 1/d}\big] : L/ \mu	\in 2\NN+1 \Big\}.
\ee
Thus, $\beta^{\lfloor{(M_{L}-2\delta_+)/{R_0}}\rfloor^d} \le \beta^{[M_{L}/(2R_{0})]^{d}}
	\le \beta^{(\log L)^{1+d\varepsilon}}$ for all sufficiently large $L$,
and, given any $p>0$, we can find $\cL=\cL(d,\eps,\beta,D',p)$ such that for $L\geq \cL$ we have
\be
	\label{ilseprob3}
	\P\Big( \inf \sigma(H_{\omega,t_{S},x,L})\geq E_0+ \mathcal{E}_{L} \Big) \geq 1-L^{-pd}
\ee
with $\mathcal{E}_{L} := u_{-}C_{UC}\big(4R_0(\log L)^{\epsilon + 1/d}\big) \le  u_-C_{UC}(M_{L})$ taking the explicit form as stated in the lemma.
This proves \eqref{ilse}.

In order to obtain the claim on good scales, we assume that $\omega$ is in the event of \eqref{ilseprob3} from
now on and let $E\in[E_0,E_0+\mathcal{E}_{L}/2]$.
Then, for every fixed $\zeta\in (0,1)$,
\be\label{ct1prime}
	\norm{R_{\omega,t_{S},x,L}(E)}\leq 2/\mathcal{E}_{L} \le e^{L^{1 - \zeta}}
\ee
for all $L\ge \tilde{\mathcal{L}} = \tilde{\mathcal{L}}(d,\eps,\beta,D',p)$, uniformly in $t_{S} \in [0,1]$.
The last inequality in \eqref{ct1prime} follows from \eqref{KLchoice} and
\eqref{c:ucp}.
Let also
$z,y\in\L_L(x)$ with $\norm{z-y}>L/100$. Then, the Combes--Thomas estimate (Theorem \ref{t:CT}) gives
\be\label{ct1}
	\norm{\Chi_z R_{\omega,t_{S},x,L}(E)\Chi_y}\leq \frac{8}{3\,\mathcal{E}_{L}}e^{-(\mathcal{E}_{L}/32)^{1/2}\norm{z-y}}
\ee
for all $E\in[E_0,E_0+\mathcal{E}_{L}/2]$ and all $L > 200 \sqrt{d}$, uniformly in $t_{S} \in [0,1]$.
We need to absorb the factor $1/\mathcal{E}_{L}$ in the exponent of the exponential.
This is only possible if
\begin{equation}
	\label{restr-d}
 	\frac{4}{3}\, \Big(\varepsilon + \frac{1}{d}\Big) < 1.
\end{equation}
It is here where the restriction to $d\ge 2$ arises. So we choose $\varepsilon\in (0, \frac34 - \frac1d)$ and obtain
\be\label{ct1doubleprime}
	\norm{\Chi_z R_{\omega,t_{S},x,L}(E)\Chi_y}\leq e^{-(\mathcal{E}_{L}/64)^{1/2}\norm{z-y}}
\ee
for all $E\in[E_0,E_0+\mathcal{E}_{L}/2]$ and all $L\ge \tilde{\mathcal{L}}$.
Introducing $m_L:=\sqrt{\mathcal{E}_L/64}$, we see that
the scales $L\geq \max\{\mathcal L, \tilde{\mathcal{L}}\}$ are $(E,m_{L},\zeta,p)$-good  for the operator with free sites for all energies $E\in[E_0,E_{0}+ \mathcal{E}_L/2]$, uniformly in $t_{S} \in [0,1]^{S}$.

\smallskip

\noindent\emph{Proof of Part} (ii).\quad
In the case $d=1$ we proceed as in Part (i) up to \eqref{eq:first-transformed}.  From \eqref{eq:gamma-loc} we see that
\be\label{eq:cubes} \L_{\delta_-}(M\gamma_k(\omega))\subset \L_M(Mk)=:\L_M(\kappa),\quad \forall \kappa\in M\ZZ\cap\L_{\ell}, \ee
and the potential satisties
\be W_{\omega}=\sum_{k\in\hat \L_\ell}u(\boldsymbol\cdot\,-M\gamma_k(\omega))\geq u_-\Chi, \ee
where $\Chi$ is the characteristic function of $\cup_{k\in\hat \L_\ell}\L_{\delta_-}(M\gamma_k(\omega))$. By conveniently extending the finite family of cubes $\{\L_{\delta_-}(M\gamma_k(\omega))\}_{k\in\hat \L_\ell}$ to an infinite sequence satisfying \eqref{eq:cubes} for all $\kappa\in M\ZZ$, the potential $W_{\omega}$ can be seen as a potential satisfying the lower bound $W_\omega\geq u_- \Chi$, where, in an abuse of notation, $\Chi$ is the characteristic function of a reunion of cubes in $\RR$ that satisfies the hypotheses of Theorem \ref{t:lift}. Therefore we can apply this Theorem with $W=W_\omega$, $s=\delta_-$, $C_{-} = u_{-}$,
$x=0$, $t=1$ and with $M\gamma_k(\omega)$ playing the role of the centres $\gamma_k$ there. This implies
\be \inf \sigma(H_{D-x,0,L}+W_\omega) \geq \inf \sigma(H_{D-x,0,L})+u_-\widetilde C_{UC}(M) \ee
with a probability given by \eqref{A-prob} with $d=1$. Therefore, we obtain the equivalent of the inequality in Eq. \eqref{lowb} in dimension $d=1$, but with $C_{UC}(M)$ replaced by
\be \label{cuctilde} \widetilde C_{UC}(M) :=\frac{\delta_-}{4M}\,e^{-CM}, \ee
where the constant $C>0$ depends only on $u, \norm{V}_\infty, \sup_{\omega\in\Omega}\norm{W_\omega}_\infty$, and $E_0$, but not on $M$.

Next, we choose
\be M= \widetilde M_{L}:=\max\Big\{ \mu \in \big(R_{0}+ 2\delta_{+},\frac{4R_0(1+p)\log L}{\abs{\log \beta}} \big] : L/ \mu	\in 2\NN+1 \Big\}, \ee
which ensures that, given $p>0$,
\begin{align}
	(L/\widetilde M_L)\beta^{\lfloor(\widetilde M_L-2\delta_+)/R_0\rfloor}
	& \leq  (L/\widetilde M_L)\beta^{[\widetilde M_L/2R_0]}  \notag\\
	& \leq  (L/\widetilde M_L)\beta^{(1+p) \log L/|\log\beta|}\leq L^{-p}
\end{align}
for all $L$ large enough.
Then, there exists $\cL=\cL(u,\beta,D',p)$ such that for all $L\geq \cL$, we have \eqref{ilseprob3} for $d=1$ with
$\widetilde{\mathcal{E}}_{L} := u_{-} \widetilde{C}_{UC}(\widetilde M_{L})$ taking the form as in the claim.

Finally, \eqref{ct1prime} and \eqref{ct1} hold for $d=1$ and $\mathcal{E}_{L}$ replaced by $\widetilde{\mathcal{E}}_{L}$.
The restriction $\beta < e^{-\tilde C}$ is required in order to pass from \eqref{ct1} to \eqref{ct1doubleprime}, with $\tilde C$ as in the claim of the Lemma. Thus, the property of good scales follows as above, given the constraint on $\beta$.
%
%
\end{proof}

\begin{remarks}
\item
	As we have seen, the proof of localisation using the MSA gives the existence of an energy $E_*>E_0$ such that the operator exhibits almost surely dynamical localisation in the interval $[E_0,E_*]$. In the ergodic setting with regular random variables it was shown in \cite{GK2} how $E_*$ depends on the parameters of the model. This was later used in \cite[Thm.~3.9]{GMRM} for a Delone--Anderson model, to show how this energy depends on the parameters of the underlying Delone set. It would be interesting to know if one can obtain similar information in our case. However, the proof of \cite{GK2} relies on Wegner estimates and is therefore restricted to regular random variables.
\item
	Note that even if it was $M$ instead of $M^{4/3}$ appearing in \eqref{c:ucp} in the exponent of $C_{UC}(M)$
	 -- as it is conjectured, see for example \cite{DKW17} -- the proof of Lemma~\ref{t:ilse}(i) would involve the condition $\frac{1}{d}+\eps<1$ instead of \eqref{restr-d}. Therefore the restriction to $d\geq 2$ could still not be relaxed in this part.
	\end{remarks}

\section{Topological properties of Delone--Bernoulli operators}\label{s:delops}

In this section we apply Theorem \ref{t:dynloc} to obtain topological information on the set of Delone sets for which the associated Delone operators exhibit localisation.

\subsection{Denseness of Delone sets exhibiting localisation}
\label{s:pp}

In this section we show that Delone sets whose associated Delone operators exhibit dynamical localisation at
low energies are dense in $\DD$. Even though this is a purely deterministic statement, our proof is not and
relies on an application of Theorem~\ref{t:dynloc}.

\begin{thm}
	\label{p:conv}
 	Let $D\in\DD$ be a Delone set and let $E_{0}$ be the infimum of the spectrum
  \eqref{infspec-H-D} of the associated Delone operator \eqref{delop}.
  Then there exists an energy $E_* > E_{0}$ and a sequence $(D_{L})_{L\in\NN} \subset \DD$ of Delone sets such that
  $\lim_{L\to\infty} D_{L} =D$ in the topology of $\DD$ and the Delone operators $H_{D_L}$ exhibit
  dynamical localisation in $[E_0, E_*]$ for all $L\in\NN$.
\end{thm}

\begin{rem}
 We stress that the energy $E_{*}$ in the above theorem does not depend on $L$. This signals how weak the notion of convergence in $\DD$ is.
\end{rem}

\begin{proof}[Proof of Lemma~\ref{p:conv}]
  It suffices to prove the lemma for all $\ell\ZZ^{d}$-periodic Delone
  sets with $\ell\in\mathbb{N}$ arbitrary, because these are dense in
  $\DD$, as we show now: given an arbitrary Delone set $D$, define $D_{\ell}$ as
  the periodic tiling of $\RR^{d}$ by all $\ell\ZZ^{d}$-translates of
  $D\cap \Lambda_{\ell}$, viz.\
  \be
  	D_\ell := \bigcup_{j\in\ell\ZZ^d} \big\{ j+ (D\cap \L_\ell)\big\}.
	\ee
  Then, $\lim_{\ell\to\infty} D_{\ell} =D$ in the topology of $\DD$, see Definition \ref{d:conv}.

  So let us assume that $D\in\DD$ is $\ell\ZZ^{d}$-periodic for some
  $\ell\in\NN$ (we have dropped the $\ell$ from the notation) and let $(D,D')$ be a Delone pair, as in Definition \ref{def:part}. Consider the
  Delone--Bernoulli operator $H_{D^{\prime\prime\omega}}$ from Definition~\ref{def:delberop}.
  If $d=1$ assume $\beta \in (0,C)$ there with the constant $C\in(0,1]$ from Theorem~\ref{t:dynloc}.
  Let $(\Lambda_{L}^{(j)})_{j\in\NN}$ be a sequence of disjoint cubes in
  $\RR^{d}$, centred around $x_{L}^{(j)} \in\ell\ZZ^{d}$ and with
  edges of length $L\in\mathbb{N}$. The events
  \begin{equation}
    A_{L,j} := \Big\{\omega\in\Omega_{D'} : D''^{\omega} \cap \Lambda_{L}^{(j)}
    = D \cap \Lambda_{L}^{(j)}\Big\}
  \end{equation}
  that all random variables $\omega_{\gamma}=0$ for $\gamma\in D' \cap\Lambda_{L}^{(j)}$ are pairwise
  independent and have strictly positive probability $\PP_{D'}(A_{L,j})=\beta^{|D'\cap\L_L^{(j)}|} \ge c_{L} >0$
  with $c_{L}$ independent of $j$. Hence, the Borel--Cantelli
  Lemma implies $\PP_{D'}(A_{L}) =1$ for the event $A_{L} :=
  \limsup_{j\to\infty} A_{L,j}$ of the joint occurrence of infinitely
  many of the events $A_{L,j}$.

  From Theorem~\ref{t:dynloc} we infer the existence of an energy $E_*>E_0$ and an event
  $\hat \Omega_{D'}\subset \Omega_{D'}$ of full probability for which the operator
  $H_{D''^\om}$ with  $\om\in \hat \Omega_{D'}$ exhibits dynamical localisation in $[E_0,E_*]$
  and all its consequences. Then, $\PP_{D'}(\hat \Omega_{D'} \cap A_{L})=1$ for all $L\in\NN$.  In particular,
  for every $L\in\NN$ there exists $\omega_{L}\in\hat \Omega_{D'}$ such
  that $(\omega_L)_{\gamma} =0$ for all points $\gamma$ lying in some cube
  $\Lambda_{L}^{(j)}$ of centre $x_{L}^{(j)}$. Setting
  \be
  	\label{f:transdel}
		D_{L} := -x_{L}^{(j)} + D''^{\omega_{L}},
	\ee
  we have that $H_{D_L}$ exhibits dynamical localisation in $[E_0,E_*]$, because $H_{D''^{\omega_{L}}}$ does,
  and a translation by a vector does not change the spectral type of the operator. Moreover,
  $D_{L} \cap \L_{L}(0) = D \cap \L_{L}(0)$ for all $L\in\NN$. The latter implies that
  $D =\lim_{L\to\infty} D_{L}$ in the topology of $\DD$, see Definition \ref{d:conv}.
\end{proof}

\subsection{Meagreness of localised configurations in probability space}
\label{s:sc}

We present two approaches in this section.
For our first approach, we recall a version of Simon's Wonderland Theorem \cite[Sect.~2]{Si}
obtained in \cite{LS06}. Let $\cS$ be the space of self-adjoint operators on $\Lp{\RR^d}$,
endowed with the strong resolvent topology $\tau_{srs}$, that is, the coarsest topology on $\cS$
such that, for every $\varphi\in \Lp{\RR^d}$, the mapping $A\ni \cS\mapsto (A+i)^{-1}\varphi \in \Lp{\RR^d}$ is
continuous. Convergence in this topology corresponds to convergence of operators in the
strong resolvent sense \cite[Sect.\ VIII.7]{RSI}.

\begin{thm}[Thm.\ 2.1 in \cite{LS06}]
	\label{t:wond}
	Let $(X,d)$ be a complete metric space, $H:X \rightarrow \cS$, $x\mapsto H_{x}$, be a continuous
	mapping and $U\subseteq \RR$ an open set. Assume that each of the following three sets
	\begin{itemize}
	\item[(i)] \quad$\{ x\in X\,:\, \sigma_{pp}(H_{x})\cap U=\emptyset \}$,
	\item[(ii)] \quad$\{ x\in X\,:\, \sigma_{ac}(H_{x})\cap U=\emptyset \}$,
	\item[(iii)] \quad$\{ x\in X\,:\, U\subseteq \sigma(H_{x}) \}$
	\end{itemize}
 	is dense in $X$. Then,
	\begin{equation}
		\big\{ x\in X\,:\, U\subseteq \sigma(H_{x}), \;\sigma_{ac}(H_{x})\cap U= \emptyset = \sigma_{pp}(H_{x})\cap U \big\}
	\end{equation}
	is a dense $G_\delta$-set in $X$.
\end{thm}

In our application of the wonderland theorem to Delone--Bernoulli operators, the following notion will be useful.

\begin{defn}
	\begin{nummer}
	\item Given an $(r,R)$-Delone set $D$ and a compact set $K\subset \RR^d$, we call the subset $D\cap K$
		the $K$-\emph{pattern} of $D$.
	\item We say that an $(r,R)$-Delone set $D$ is \emph{infinitely pattern repeating} (IPR), if
		for every compact set $K\subset \RR^d$ there exist infinitely many translates of the $K$-pattern in $D$,
		\be
			\big|\big\{ y\in \Rd :  y+(D\cap K)=D\cap (y+K) \big\}\big| = \infty.
		\ee 		
	\end{nummer}
	\label{d:supf}
\end{defn}

We remark that every periodic Delone set $D$ has the IPR property.


\begin{thm}\label{t:meagre} Assume that $D$ has the IPR property and that $u\in\cC_c^1(\RR^d)$. Assume moreover that
the following holds:
\begin{nummer}
\item there exists an energy $E'$ such that $H_D$ exhibits purely absolutely continuous spectrum in the interval $[E_0,E']$,
\item if $D'$ is a Delone set such that $(D,D')$ is a Delone pair, and $K\subset \RR^d$ is any compact set, then the spectrum of the operator $H_D+\tilde V$ with
\be\tilde V :=\sum_{\gamma\in D'\cap K} u(\,\pmb\cdot\,-\gamma)\ee
remains purely absolutely continuous in the interval $[E_0,E']$.
    \end{nummer}
Then, there exists an energy $\hat E$ such that the event $\Omega_{sc}\subset \Omega_{D'}$ of $\om$ for which $H_{D''^\om}$ exhibits purely singular continuous spectrum in the interval $(E_0,\hat E)\subset [E_0,E']$ is a dense $G_\delta$-set in $\Omega_{D'}$.
\end{thm}

The theorem is proven below. It immediately implies

\begin{cor}
 	Given the assumptions of Theorem~\ref{t:meagre}, the event $\Omega_{pp}$ of $\omega\in \Omega_{D'}$ for which $H_{D''^\om}$ has non-empty pure point spectrum in $(E_0,\hat E)$ is a meagre subset of $\Omega_{D'}$.
	In particular, the events that $H_{D''^\om}$ exhibits Anderson or dynamcial localisation in $(E_0,\hat E)$ are meagre.
\end{cor}

\begin{rem}\label{r:wond}
 	Both assumptions (i) and (ii) in the above theorem hold in the absence of a background potential, $V_{D}=0$.
	The same is true for certain periodic background potentials. For more general periodic $V_{D}$, this is conjectured, see \cite{KuV} and references therein.
	In the discrete setting, this has recently been proved in \cite{LiuOng, Liu}.
	For Delone background potentials $V_{D}$, even the implication ``(i) $\Rightarrow$ (ii)'' is not known, but it should hold under reasonable assumptions.
 \end{rem}

In view of the previous remark, we now set $D=\emptyset$ (not a Delone set!) and $D'=\ZZ^d$.
Theorem \ref{t:meagre}, which also holds in this case, then extends a result obtained by Simon \cite{Si}
for certain regular random potentials to the case of random potentials with singular probability distributions.
We formulate this as

\begin{cor}
	There exists $\hat{E} >0$ such that the continuum Bernoulli--Anderson model, as given by
	Definition~\ref{def:delberop} with $D=\emptyset$ and $D'=\ZZ^d$,
	generically exhibits singular continuous spectrum in the interval $[0, \hat{E}]$
	above the bottom of the spectrum.
\end{cor}

\begin{proof}[Proof of Theorem~\ref{t:meagre}]
 Consider $X=\Omega_{D'}$ endowed with the metric
 	\begin{equation}
		\label{om-met}
 		\dist_{D'}(\omega,\omega') := \sum_{\gamma\in D'} 2^{-\|\gamma\|} |\omega_{\gamma} - \omega'_{\gamma}|
	\end{equation}
 	for $\omega,\omega'\in\Omega_{D'}$. This metric is compatible with the product topology of $\Omega_{D'}$ and
	renders $\Omega_{D'}$ a complete metric space.
	The mapping $H:(\Omega_{D'},\dist_{D'})\rightarrow (\cS,\tau_{srs})$, given by $H_{\omega}
	:= H_{D''^{\omega}}: =-\Delta+
		V_D+ \sum_{\gamma\in D'}\om_\gamma u(\,\cdot\,-\gamma)$, is continuous by the same arguments as in \cite[Eqs.\ (2.28) -- (2.32)]{GMRM}.
		
	The fact that the interval $[E_0,E']$ is contained in the spectrum of $H_D$ and the IPR property imply by Theorem \ref{t:asspec} that $[E_0,E']$ is also contained in the spectrum of $H_{\omega}$ for $\om$ in a set of full probability in $\Omega_{D'}$, which is therefore dense in $\Omega_{D'}$ by Lemma~\ref{l:densesupp}.
	
	The configurations $\om$ that are supported on a compact set, that is, $\om_\gamma\neq 0$ for at most finitely many $\gamma$, are dense in $\Omega_{D'}$. Moreover, the operators $H_{\omega}$ associated to these $\om$ exhibit purely absolutely continuous spectrum in the interval $[E_0,E']$ by hypothesis (ii). Thus, they do not have any pure point spectrum in this interval.
	
	Next, by Theorem~\ref{t:dynloc}, there exists an energy $E_*$ such that
$H_{\omega}$ exhibits localisation in $[E_0,E_*]$ for every $\omega\in\Omega_{MSA}$ with $\PP(\Omega_{MSA}) =1$. Thus, $\Omega_{MSA}$ is dense in $\Omega_{D'}$ by  Lemma~\ref{l:densesupp}.
Since the event that $H_{\omega}$ exhibits only pure point spectrum in $[E_0,E_*]$ -- and therefore no absolutely continuous spectrum there -- contains $\Omega_{MSA}$, this event is also dense in $\Omega_{D'}$.
Therefore, we can apply Theorem \ref{t:wond} in the interval $(E_0,\hat E)$, with $\hat E: =\min \{E',E_*\}>E_0$ and obtain that the event  $\{\om\in\Omega_{D'}: \sigma(H_{\omega}) \cap (E_0,\hat E)=\sigma_{sc}(H_{\omega})\}$ is a dense $G_\delta$-set.
\end{proof}

\begin{rem}
\label{random-displ-m}
An analogous result as in Theorem \ref{t:meagre} can be obtained for the Random Displacement Model,
see (ii) after Definition~\ref{d:conv} in the introduction.
In the proof one has to replace the localisation statement of Theorem~\ref{t:dynloc} by the one in \cite{KLNS10}.
We remark that the proof of localisation for this model \cite{KLNS10} relies heavily on the underlying lattice
structure and therefore cannot be extended to the more general Delone setting we consider here.
\end{rem}

Using the multiscale analysis \emph{directly}, we obtain an alternative theorem, which gives a weaker statement
than Theorem \ref{t:meagre} but holds under more general assumptions.

\begin{thm}\label{t:md}
	Let $E_{*}$ and $\Omega_{MSA}$ be given by Theorem~\ref{t:dynloc}.
 	Assume that $H_{D}$ exhibits purely absolutely continuous spectrum
	in the interval $[E_0,E_{*}]$.
	Then, the event $\Omega_{MSA}$ is a meagre dense set in $\Omega_{D'}$.
\end{thm}

\begin{proof}
	Density of $\Omega_{MSA}$ comes from the fact that $\mathbb P(\Omega_{MSA})=1$ and Lemma \ref{l:densesupp}.
	It remains to show meagreness.
	The event $\Omega_{MSA}$ is constructed in \cite[Thm.\ 7.2]{GKuniv}, using the output of the
	multiscale analysis \cite[Thm.\ 6.1]{GKuniv}.
	It is contained in the event $\cU_\infty$ of $\omega\in\Omega$ for which $H_{D''^\om}$ exhibits
	Anderson localisation, see \cite[Thm.\ 7.1]{GKuniv} for the definition or Theorem~\ref{t:Uinfty}.
	Thus, the claim follows from the stronger statement that $\cU_\infty$ is a meagre set in $\Omega_{D'}$.

	In order to see that $\cU_\infty$ is meagre, that is, a countable union of nowhere dense sets, we appeal to
	\eqref{uinfty} and show that each of the sets $\bigcap_{k=n}^{\infty}\cU_{0,L_k}$, $n\in\NN$, is closed
	in $\Omega_{D'}$ and has empty interior. Closedness follows from Lemma~\ref{l:closed} in the appendix.
	To show empty interior we consider the even bigger event
	\be
		\Omega_{pp} := \big\{ \omega\in\Omega: H_\om \mbox{ exhibits only pure point spectrum in }
			[E_{0}, E_{*}]\big\}.
	\ee
	This event has empty interior in the product topology of $\Omega_{D'}$: indeed, take
	$\omega\in\Omega_{pp}$ and an arbitrary neighbourhood of $\omega$ in $\Omega_{D'}$. Then, this
	neighbourhood contains an element $\omega' = (\omega'_\gamma)_{\gamma\in D'}$ with $\omega'_{\gamma}=0$
	for all $\gamma$ outside a compact set of $\RR^{d}$. The operator $H_{\omega'}$ has absolutely continuous
	spectrum in $[E_{0},E_{*}]$ because it is a perturbation of $H_D$ by a compactly supported potential.
	Therefore, $\omega'\notin \Omega_{pp}$, which shows that $\Omega_{pp}$ has empty interior.
	Thus, for every $n\in\NN$, $\bigcap_{k=n}^{\infty}\cU_{0,L_{k}} \subseteq \cU_\infty\subseteq
	\Omega_{pp}$ has empty interior in the product topology of $\Omega_{D'}$.
\end{proof}

\begin{rem}
	We note that Theorem \ref{t:md} also holds for the case of Delone--Anderson operators with H\"older-continuous probability distributions. To see this, it is enough to mention that  the steps outlined in Appendix \ref{A:meagMSA} hold for all non-degenerate probability distributions.
\end{rem}

\begin{appendix}
\section{On the spectral infimum of Delone--Anderson operators}\label{specinf}

%

Here, we consider a generalisation of the Delone--Bernoulli operator $H_{\omega} :=H_{D''\omega}$
from Definition \ref{def:delberop} to the case of an arbitrary single-site distribution $\mu$ with
support in more than one point of $\RR$. We refer to this generalisation as the \emph{Delone--Anderson operator}.

\begin{thm}\label{t:asspec}
  Let $D$ and $D'$ be Delone sets, with $D$ satisfying the IPR property. Let
  $\Omega_{D'} \ni \omega \mapsto H_{\omega}$ be a Delone--Anderson operator with $0\in\supp(\mu)$.
 	Then, there exists $\widetilde \Omega_{D'}\subseteq \Omega_{D'}$ with $\P_{D'}(\widetilde \Omega_{D'})=1$
	such that
 	\be
		\sigma(H_D)\subseteq \sigma(H_{\omega}) \quad\mbox{for all }\, \om\in\widetilde \Omega_{D'}.
	\ee
	Moreover, $\sigma(H_\omega) \cap \sigma(H_{D}) = \sigma_{ess}(H_{\omega}) \cap \sigma(H_{D})$
	for all $\omega\in\Omega_{D'}$.
\end{thm}


The proof of Theorem~\ref{t:asspec} relies on a Weyl-sequence argument. Among others, the IPR property ensures
that the members of a Weyl sequence can be chosen to be mutually orthogonal.
This argument works because $D\cup D'$, despite not necessarily being a Delone set, has no accumulation points, and therefore domain problems cannot arise.
We refer to \cite[Prop.\ 3.2]{RM13}
 for an analogous proof in the discrete setting.

\section{The quantitative unique continuation principle in dimension one}\label{A:QUCP}
In this section we recall an intermediate result obtained in the proof of \cite[Lemma 5]{KV02} (see also \cite{V96}) which implies a quantitative unique continuation principle for eigenfunctions of Schr\"odinger operators in dimension one. The latter is similar to \cite[Thm.\ 2.1 and Cor.\ 2.2]{RMV12}, but with a slightly different constant $C_{UC}$. The proof of \cite[Lemma 5]{KV02} can be extended to the setting of Delone operators since the arguments do not rely on the periodicity of the underlying structure of the potential. Due to this fact, this result has been extended to non translation invariant settings, for example, to metric graphs, see \cite{HV07,GHV08}.

\begin{lem}\label{l:lemkv}\cite[Lemma 5]{KV02} Let $0<s<L$ be fixed, and consider the interval
	$\L_L(x)=(x-L/2,x+L/2)$, with $x\in\RR$. Let $H=-\Delta+V$ be a Schr\"odinger operator with bounded potential $V$, and consider its finite-volume restriction $H_{x,L}=-\Delta_{x,L}+V_{x,L}$ to $\Lp{\L_L(x)}$ with Dirichlet boundary conditions. Let $\varphi$ be an eigenfunction of $H_{x,L}$ with eigenvalue $E\in\RR$. Let $k\in\L_L(x)$ be such that $\L_s(k)\subset \L_L(x)$. Then, there exists a constant $C_{s,V,E}>0$, such that for any $y \in\L_{L}(x)$ with $\L_s(k+y)\subset \L_L(x)$ we have
\be\label{lemkv} \norm{\varphi}^2_{\L_s(k+y)}\leq e^{ C_{s,V,E}\abs{y}}\norm{\varphi}^2_{\L_s(k)},
 \ee
where $\norm{\varphi}_{\L}$ denotes the $L^{2}$-norm of $\varphi$ restricted to $\L$. The constant $C_{s,V,E}>0$ is known explicitly
\be\label{cucp} C_{s,V,E}=2{\left( \frac{2\cdot 18^2}{s^2}+ 2s^2 \norm{V-E}^2_\infty\right)^{1/2}}. \ee
\end{lem}

\begin{rem}
The estimate \eqref{lemkv} is obtained from \cite[Eq.\ (16)]{KV02}. The expression in  \eqref{cucp} can be obtained from the proof \cite[Thm.\ 7.27, Eq.\ (7.60)]{GT} with the choice $p=2$, using the eigenvalue equation for $\varphi$.

\end{rem}


\begin{thm}\label{t:ucp1}Let $0<s<M<L$ be fixed. Consider the interval $\L_L(x)=(x-L/2,x+L/2)$, with $x\in\RR$. Let $H=-\Delta+V$ be a Schr\"odinger operator with bounded potential $V$, and consider its finite-volume restriction $H_{x,L}=-\Delta_{x,L}+V_{x,L}$ to $\Lp{\L_L(x)}$ with Dirichlet boundary conditions. Let $\varphi$ be an eigenfunction of $H_{x,L}$ with eigenvalue $E\in I$. Let $(\gamma_k)_{k\in M\ZZ}\subset \RR$ be a sequence such that
\be\L_s(\gamma_k)\subset \L_M(k) \quad\mbox{for each}\,k\in M\ZZ.
\ee
Then,
\be\label{ucp1}\sum_{k\in M\ZZ:\,\gamma_k \in \L_{L-s}(x) } \norm{\varphi}^2_{\L_s(\gamma_k)} \geq C_{UC}^{(V,E)}(M) \norm{\varphi}_{\L_L(x)}^2,
\ee
where
\be C_{UC}^{(V,E)}(M)=\frac{s}{4M}e^{-2C_{s,V,E}M}\ee
with the constant $C_{s,V,E}$ given in \eqref{cucp}.
\end{thm}

\begin{proof}
Consider the index set defined as
\be\label{index} \cJ_{x,L,M}:=\cJ_{x,L,M}^{(1)} \cup \cJ_{x,L,M}^{(2)}\ee
  with
 \be \cJ_{x,L,M}^{(1)}:=\{ k\in M\ZZ:\,\L_M(k)\subset \L_L(x)\},
\ee
\be \cJ_{x,L,M}^{(2)}:=\left\{x-\frac{L}{2}+\frac{M}{2},x+\frac{L}{2}-\frac{M}{2}  \right\}.
\ee
Then, we can write
 \be \L_L(x)=\left(\bigcup_{\kappa \in \mathcal J_{x,L,M}}\overline{\L_M(\kappa)}\right)^{\rm int}. \ee
  Note that in this covering of $\L_L(x)$ there might be an overlap of positive Lebesgue measure near the boundary of $\L_L(x)$, therefore,
\be\label{ineq1}  \|\varphi\|^{2}_{\L_{L}(x)} \leq \sum_{\kappa\in \mathcal J_{x,L,M}} \|\varphi\|^{2}_{\L_M(\kappa)}. \ee
In turn, we write
 \be \L_M(\kappa)=\left(\bigcup_{y \in \mathcal J_{\kappa,M,s}}\overline{\L_s(y)}\right)^{\rm int} \ee
with the index set $\mathcal J_{\kappa,M,s}$ defined analogously to $\mathcal J_{x,L,M}$, see \eqref{index}. Then, we have
\be\label{ineq2} \|\varphi\|^{2}_{\L_M(\kappa)} \leq \sum_{y \in \mathcal J_{\kappa,M,s}}
\|\varphi\|^{2}_{\L_s(y)} . \ee
This yields
\begin{align}  \|\varphi\|^{2}_{\L_{L}(x)}  & \leq \sum_{\kappa\in \mathcal J_{x,L,M}}\sum_{y \in \mathcal J_{\kappa,M,s}} \|\varphi\|^{2}_{\L_{s}(y)} \\
& = \sum_{\kappa\in \mathcal J_{x,L,M}^{(1)}}   \sum_{y \in \mathcal J_{\kappa,M,s}} \|\varphi\|^{2}_{\L_{s}(y)}
		+ \sum_{\kappa\in \mathcal J_{x,L,M}^{(2)}} \sum_{y \in \mathcal J_{\kappa,M,s}} \|\varphi\|^{2}_{\L_{s}(y)}.
\label{ineq4}
\end{align}
For the last expression, we can use Lemma \ref{l:lemkv} to estimate the norm $\norm{\varphi}^2_{\L_s(y)}$ from above in terms of $\norm{\varphi}^2_{\L_s(\gamma_k)}$ with some $\gamma_k$, ${k\in M\ZZ}$, as long as $\L_s(\gamma_k)\subset \L_L(x)$.
 For $k=\kappa\in \cJ_{x,L,M}^{(1)}$ there is a unique element $\gamma_k$ that satisfies $\L_s(\gamma_k)\subset \L_M(k) \subset\L_L(x)$ and therefore to each $y \in \mathcal J_{\kappa,M,s}$ we can associate this element $\gamma_k$ . Then, we can apply Lemma \ref{l:lemkv} and \eqref{lemkv} holds in the form  
\be\label{eq:upb1} \norm{\varphi}^2_{\L_s(y)}\leq e^{ C_{s,V,E}\abs{y-\gamma_\kappa}} \norm{\varphi}^2_{\L_s(\gamma_\kappa)}, \ee
where $\abs{y-\gamma_\kappa}<M$ by construction.

Now, for the second sum in \eqref{ineq4}, for $\kappa\in \cJ_{x,L,M}^{(2)}$ there is not a one-to-one
correspondence with elements of  $(\gamma_k)_{k\in M\ZZ}$ such that $\L_s(\gamma_\kappa)\subset \L_L(x)$. Indeed, due to the fact that $L$ is not necessarily a multiple of $M$, $\L_M(\kappa)$ with $\kappa\in \cJ_{x,L,M}^{(2)}$ might intersect totally, partially, or none of the intervals $\L_s(\gamma_k)$. If there is an interval $\L_s(\gamma_k)$ fully contained in $\L_M(\kappa)$, then we can directly apply Lemma \ref{l:lemkv} as before. However, if this is not the case,
we can still use Lemma  \ref{l:lemkv} if we consider a $\gamma_k$ that is contained in the neighbouring interval $\L_M(\kappa')$ with $\kappa'\in \cJ_{x,L,M}^{(1)}$. Therefore, \eqref{eq:upb1} holds with $\abs{y-\gamma_\kappa}<2M$, that is, with a constant twice as big.

To describe this rigorously, let us denote by $k_l:=\min \cJ_{x,L,M}^{(1)}-M \in M\ZZ$ and $k_r:=\max \cJ_{x,L,M}^{(1)}+M\in M\ZZ$ the elements $k\in M\ZZ$ that do not lie in $\L_L(x)$ and that are the closest to $\L_L(x)$ to the left and the right, respectively. We define the map $\tau: \cJ_{x,L,M} \rightarrow  \cJ_{x,L,M}^{(1)}\cup\{k_l,k_r\}$ as follows:
\be
\tau(\kappa):=\begin{cases}
\kappa &
 \text{if } \kappa\in \cJ_{x,L,M}^{(1)}\\ \min \cJ_{x,L,M}^{(1)} & \text{if}\, \kappa=x-\frac{L}{2}+\frac{M}{2}\, \text{and}\, \L_s(\gamma_{k_l})\nsubseteq \L_L(x)
 \\
  k_l & \text{if}\, \kappa=x-\frac{L}{2}+\frac{M}{2}\, \text{and}\, \L_s(\gamma_{k_l})\subset \L_L(x)\\
 \max \cJ_{x,L,M}^{(1)} &  \text{if}\, \kappa = x+\frac{L}{2}-\frac{M}{2}\, \text{and}\, \L_s(\gamma_{k_r})\nsubseteq \L_L(x)\\
 k_r &  \text{if}\, \kappa = x+\frac{L}{2}-\frac{M}{2}\, \text{and}\, \L_s(\gamma_{k_r})\subset \L_L(x)
 \end{cases}.
\ee
With this construction, for every $y\in \cJ_{\kappa,M,s}$ we can write $y= \gamma_{\tau(\kappa)}+z$ with some $z\in \RR$ with $\abs{z}<2M$. Note also that since the map $\tau$ is not injective, the sequence $(\gamma_{\tau(\kappa)})_{\kappa\in \cJ_{x,L,M}}$ might contain twice the elements $\gamma_{\min \cJ_{x,L,M}^{(1)}}$ and $\gamma_{\max \cJ_{x,L,M}^{(1)}}$. We have that $\L_s(\gamma_{\tau(\kappa)}+z)\subset \L_L(x)$ for all $\kappa \in \cJ_{x,L,M}$ and therefore we can apply Lemma \ref{lemkv} to obtain
\begin{equation} \label{ineq3} \|\varphi\|^{2}_{\L_{s}(y)}= \|\varphi\|^{2}_{\L_{s}(\gamma_{\tau(\kappa)}+z)}
	 \leq e^{C_{s,V,E}\abs{z}} \|\varphi\|^{2}_{\L_{s}(\gamma_{\tau(\kappa)})}
  \leq e^{2C_{s,V,E}M} \|\varphi\|^{2}_{\L_{s}(\gamma_{\tau(\kappa)})} .
 \end{equation}
Plugging this into \eqref{ineq4} gives
\be \label{B-pre-last}
	 \|\varphi\|^{2}_{\L_{L}(x)} \leq e^{2C_{s,V,E}M} \;\frac{2M}{s}  \sum_{\kappa\in \mathcal J_{x,L,M}}
	 \|\varphi\|^{2}_{\L_{s}(\gamma_{\tau(\kappa)})}, \ee
where we used the fact that $ \sharp \cJ_{\kappa,M,s} < \frac{2M}{s}$. Next, note that
\be \label{B-last} \sum_{\kappa\in \mathcal J_{x,L,M}} \|\varphi\|^{2}_{\L_{s}(\gamma_{\tau(\kappa)})}
\leq 2\sum_{k\in M\ZZ:\,\gamma_k \in \L_{L-s}(x)}  \|\varphi\|^{2}_{\L_{s}(\gamma_{k})} ,
\ee
because the terms $\|\varphi\|^{2}_{\L_{s}(\gamma_{\tau(\kappa)})}$
with ${\tau(\kappa)}\in\{\min \cJ_{x,L,M}^{(1)},\max \cJ_{x,L,M}^{(1)}\}$ appear at most twice.  The desired result
follows from \eqref{B-pre-last} and \eqref{B-last}.
\end{proof}

From Theorem \ref{t:ucp1} we can obtain a result analogous to \cite[Theorem 4.9(a)]{RM12}. This quantifies the lifting of the spectral infimum produced by adding a Delone potential to an operator of the form $H=-\Delta+V_0$ with $V_0$ a real-valued, bounded, and measurable potential.

\begin{thm}\label{t:lift}
Let $0<s<M$, and $(\gamma_k)_{k\in M\ZZ}\subset \RR$ be a sequence such that
\be\L_s(\gamma_k)\subset \L_M(k) \quad\mbox{for each}\,\,\,k\in M\ZZ.  \ee
We denote the characteristic function of $\bigcup_{k\in M\ZZ} \L_{s}(\gamma_k)$ by $\Chi$. Let $V_0,W$ be real-valued, bounded, and measurable functions, and assume moreover that $W\geq C_-\Chi$, with $C_->0$. For $L\geq1$, $L>M$, and $x\in\RR$ we write $\lambda^{x,L}(t) := \inf \sigma(H_{x,L}(t))$ where $H_{x,L}(t)$ is the restriction of $H=-\Delta+V_0+tW$ to $\rm{L}^2(\L_{L}(x))$ with Dirichlet boundary conditions. Then,
\be\label{eq:lift} \forall t\in(0,1]:\,\, \lambda^{x,L}(t)\geq \lambda^{x,L}(0)+tC_- C_{UC}(M), \ee
where the constant $C_{UC}(M)$ is uniform in $t\in(0,1]$.
More precisely,
\be C_{UC}(M)=\frac{s}{4M}e^{-2C_{s,V,I}M}\ee
 with $C_{s,V,I}=\sup\{C_{s,V_0+tW,E}:\,t\in(0,1],E\in I\}$, where $I\subset \RR$ is a compact interval containing $\{\lambda^{x,L}(t):\, t\in(0,1],L\geq 1\}$, and $C_{s,V_0+tW,E}$ is given in \eqref{cucp}.
\end{thm}
\begin{proof}
 Denoting by $\varphi_t$ the normalized ground state of $H_{x,L}(t)$, we have
  \begin{align} \lambda^{x,L}(t)= \angles{\varphi_t,H_{x,L}(t)\varphi_t}=  & \angles{\varphi_t,H_{V_0,x,L} \varphi_t}+t\angles{\varphi_t,W\varphi_t}\\ & \geq \lambda^{x,L}(0)+ t C_- \angles{\varphi_t,\Chi \varphi_t}, \label{eq:liftt}\end{align}
where $H_{V_0,x,L}$ is the Dirichlet restriction of $-\Delta +V_0$ to $\rm{L}^2(\L_{L}(x))$.
Since $\varphi_t$ is a normalized eigenfunction of $H_{x,L}(t)$  with eigenvalue  $\lambda^{x,L}(t)$, we can apply Theorem \ref{t:ucp1} with $V=V_0+tW$, $E=\lambda^{x,L}(t)$ there,  and obtain, for every $t\in(0,1]$
\be\label{eq:ucpt} \angles{\varphi_t,\Chi \varphi_t} \geq \sum_{k\in M\ZZ:\gamma_k\in \L_{L-s}(x)}\norm{\varphi_t}^2_{\L_s(\gamma_k)} \geq C_{UC}^{(V,E)}(M), \ee
where
\be C_{UC}^{(V,E)}(M)= \frac{s}{4M}e^{-2C_{s,V,E}M}, \ee
and $C_{s,V,E}$ is given in \eqref{cucp}.
Note that the family of potentials $\{V=V_0+tW:\,t\in(0,1]\}$ is uniformly bounded, and therefore, there exists a compact interval $I\subset \RR$, such that $\{\lambda^{x,L}(t): t\in(0,1],L\geq1\}\subset I$.
By taking
\be C_{UC}(M):=\frac{s}{4M}\, e^{-2C_{s,V,I}M} \ee
with
\be C_{s,V,I}:= \sup \{  C_{s,V_0+tW, E} :\, E \in I,t\in(0,1] \}, \ee
we obtain a lower bound for $ C_{UC}^{(V,E)}(M)$  that is uniform in $t\in (0,1]$, since $\frac{s}{4M}<1$.
Combining this with \eqref{eq:liftt} and \eqref{eq:ucpt}, yields
\be \lambda^{x,L}(t) \geq \lambda^{x,L}(0)+ t C_- C_{UC}(M).\ee
\end{proof}

\begin{rem} From \cite[Remark 4.8]{GKuniv} we see that if the constant in the  unique continuation principle is of the form $C_{UC}(M) \sim M^{-CM^\gamma}$ for some $\gamma>0$, then in order to perform the multiscale analysis the following must hold,
\be \gamma< \frac{1+\sqrt{3}}{2},\quad \gamma-1<p<\frac{1}{2\gamma}, \ee
where the exponent $p>0$ is introduced in Definition~\ref{def:gscale}.
From Eq. \eqref{ucp1} we see that in dimension $d=1$ the unique continuation principle holds with $\gamma=1$. This implies that the only restriction on how small $p$ must be to perform the multiscale analysis for Bernoulli potentials is $0<p<\frac{1}{2}$, which is an improvement compared to the restriction $p\in (\frac{1}{3},\frac{3}{8})$ for Bernoulli random variables in $d\geq 2$.
\end{rem}

\section{The topological structure of configurations exhibiting localisation} \label{A:meagMSA}


In this section we present auxiliary results for the proof of Theorem~\ref{t:md}.
We will recall results from Sections~6 and~7 in \cite{GKuniv} and rephrase \cite[Thm.\ 7.1]{GKuniv}
as follows:

\begin{thm}\label{t:Uinfty}
	Let $\Omega_{D'} \ni \omega \mapsto H_\om$ be a Delone--Anderson operator
	on $\Lp{\RR^d}$ with either H\"older continuous or Bernoulli distributed random variables.
	Let $\cI\subset\RR$ be a bounded open interval, for which there is a scale $\cL_1$ such that for all $L\geq \cL_1$  there exist  events $\cU_{x_0,L}$, defined in \cite[Thm.\ 6.1]{GKuniv}, for all $x_0\in\RR^d$. Then, for a sequence of scales $\{L_k\}_{k=0,1,...}$  with $L_0\geq\cL_1$, $L_{k+1}=2L_k$ for $k=1,2,...$, and $x_0=0$, the event $\cU_\infty\subseteq \Omega_{D'}$, defined by
\be\label{uinfty} \cU_\infty :=  \bigcup_{n=1}^\infty \bigcap_{k=n}^\infty \cU_{0,L_k} ,\ee
is such that $\P\pa{\cU_\infty}=1$ and $H_\omega$ exhibits Anderson localisation in $\cI$ for all $\om\in\cU_\infty$.
\end{thm}

We recall that the notion of Anderson localisation was introduced above Theorem~\ref{t:dynloc}.

Given that the output of the multiscale analysis involves estimates that are uniform with respect
to the centre of a box, see e.g.\ Definition \ref{def:gscale}, it follows that \cite[Thm.\ 7.1]{GKuniv},
which is valid for  alloy-type random potentials centred on the points of $\ZZ^{d}$,
also applies to the case of operators we consider here.

When studying events from a topological perspective, the following will be helpful.
We write $\Omega := \Omega_{D'}$ for brevity.

\begin{lem}\label{l:densesupp}
	Let $S:= \supp(\mu)$ be compact in $\RR$ and let $\Omega:= \raisebox{-1pt}{\Large$\times$}_{\!D'} S$.
	Then, every set that has full $\P$-measure in $\Omega$ is dense in $\Omega$.
\end{lem}

\begin{proof}
	A set of full measure is dense in the support of the measure \cite[Prop.\ 4.1.17]{KaHa}.
	Therefore, the lemma will follow from $\supp(\P)=\Omega$.
	
	To show this, let $\omega\in \Omega$, $\varepsilon>0$ and consider the open ball
	$B_{\varepsilon}(\omega) := \{\omega'\in\Omega: \dist_{D'}(\omega,\omega') <\varepsilon\}$
	of radius $\varepsilon$ centred about $\omega$. Here, we use the metric \eqref{om-met} on
	$\Omega$. Since $S$ is bounded there exists $K \subset\RR^{d}$ compact and $\delta>0$ such that
	$Z:= \raisebox{-1pt}{\Large$\times$}_{\gamma\in D'} Z_{\gamma} \subseteq B_{\varepsilon}(\omega)$, where
	$Z_{\gamma} := (\omega_{\gamma}-\delta, \omega_{\gamma}+\delta)$ for $\gamma\in K$ and
	$Z_{\gamma} := S$ for $\gamma\in \RR^{d}\setminus K$. Thus, we have
	\be
		\P\big( B_{\varepsilon}(\omega)\big) \ge \P(Z) =  \prod_{\gamma\in K} \mu(Z_{\gamma})>0
	\ee
	because $\omega_{\gamma} \in S$ for every $\gamma\in D'$. Since $\varepsilon>0$ was arbitrary we conclude that $\omega\in\supp(\P)$.
\end{proof}

As seen in Equation \eqref{uinfty}, the events $\cU_{x_0,L}$ from \cite[Thm.\ 6.1]{GKuniv} (see Theorem \ref{UxL} below) are the building blocks of the set $\cU_\infty$ of configurations that exhibit Anderson localisation. The events $\cU_{x_0,L}$ consist of configurations $\omega$ that encode
good decay properties of generalised eigenfunctions of $H_\omega$ associated to generalised eigenvalues in $\cI$, in the sense of Theorem \ref{UxL} below. Before analysing their structure, we recall some notation: as in \cite{GKuniv}, assuming $L>1$, we fix $L_+=c_+L$, and $L_-=c_-L$ where $c_{-}\in(\frac{1}{2},1)$, and $c_+\in(2,3)$. We denote by $\mathcal F_{\L}$ the $\sigma$-algebra generated by the random variables $\omega_{\L}=(\omega_j)_{j\in D' \cap\L}$ for $\L\subset \RR^d$ a compact set.

A notion that is central in the multiscale analysis is the concept of a \emph{good} box. A box $\L_L(x)$ with $x\in\RR^d$ is called good if the norm of the resolvent $R_{\om,L}(E)$ associated to the restriction of $H_\omega$ to the box $\L_L(x)$ has good decay properties for energies $E$ in a certain interval $\cI$, i.e., it decays polynomially, or exponentially with the length of the box $L$ (see \cite[Definition 3.1]{GKuniv}). Otherwise, $\L_L(x)$ is called \emph{a bad} box.


\begin{lem} \label{l:closed}
Let $x_0\in\RR^d$ and $k\in\NN$. Then the set $ \cU_{x_0,L_k}$ is closed in the product topology of $\Omega$.
\end{lem}
\begin{proof}

Given $x_0\in\RR^d$ and $L>1$, the set $\cU_{x_0,L}$ has a cascade structure that consists of countable intersections of sets that satisfy properties (i), (ii) and (iii) in Theorem \ref{UxL} in different degrees.

In the case of H\"older continuous random variables, it is written explicitly in \cite[Eq.\ (6.97)]{GKuniv}, as follows:
 \be\label{structUxL}  \cU_{x_0,L}= \cT_{x_0,L_-} \cap\; \widehat\cM_{x_0,L}\in \cF_{\L_{L_+}(x_0)},  \ee
where the event $\cT_{x_0,L_-}$, defined in  \cite[Eq.\ (6.27)]{GKuniv}, is the set of configurations $\omega$ such that for a discrete set of energies in $\cI$, that provide a covering for $\cI$, clusters of bad boxes $\L_{\sqrt{L}}$ do not percolate to infinity. The event $\widehat\cM_{x_0,L}$ above is defined in \cite[Eq.\ (6.95)]{GKuniv} and it corresponds to configurations $\omega$ for which all boxes $\L_{\sqrt{L}}$ covering the annulus $\L_{L_+}\setminus\L_{L_-}(x_0)$ are good for the eigenvalues of $H_{\omega,x_0,L_-}$ in $\cI$.

In the case of Bernoulli random variables, the set $\cU_{x_0,L}$ is written explicitly in the proof of \cite[Prop.\ 6.12]{GKuniv}, as follows:
\be\label{structUxL2} \cU_{x_0,L}= \mathcal X_{x_0,L_-}\cap \cM_{x_0,L}, \ee
where $\cM_{x_0,L}$ is a countable intersection of the events that all boxes $\Lambda_{\frac{L}{100}}(x_0)$ covering the annulus $\Lambda_{L_+}\setminus \Lambda_{L_-}(x_0)$ are good for the eigenvalues of $H_{\omega,x_0,L_-}$ lying in a set $\mathcal I_L$, that is strictly contained in $\cI$. On the other hand,  $\mathcal X_{x_0,L_-}$ is defined in \cite[Prop.\ 6.9]{GKuniv} as
\be\label{structUxL3} \mathcal X_{x_0,L_-}= \tilde Q_{x_0,L_-}\cap \mathcal N_{x_0,L_-}, \ee
where $\mathcal N_{x_0,L_-}$, given in \cite[Lemma\ 6.11]{GKuniv}, is a countable intersection of events $\mathcal N^{(E)}_{x_0,L_+,L_-}$. The latter consist of configurations $\omega$ for which, outside a certain set that contains bad boxes, one can find a covering the annulus $\Lambda_{L_+}\setminus \Lambda_{L_-}(x_0)$ by good boxes of smaller scale.
On the other hand, the set $\tilde Q_{x_0,L_-}$, as defined in \cite[Lemma\ 6.10]{GKuniv} and \cite[Eq.\ (6.47)]{GKuniv} is given by
\be\label{structUxL4}\tilde Q_{x_0,L_-}= \mathcal T_{x_0,\hat L} \cap \mathcal Z_{x_0,\hat L}, \ee
where for a scale $\hat L$ smaller than $L$, $\mathcal T_{x_0,\hat L}$ is described after Equation \eqref{structUxL} above as a percolation event. The set $\mathcal Z_{x_0,\hat L}$ is defined in a similar way, but  involving good boxes at a much smaller scale than in $\mathcal T_{x_0,\hat L}$.

In the case of H\"older continuous random variables, we see from \eqref{structUxL} that the set $\cU_{x_0,L_k}$ is closed if both events $\cT_{x_0,L_-}$ and $\widehat\cM_{x_0,L}$ are closed. In the case of Bernoulli random variables, it can be seen from \eqref{structUxL2},\eqref{structUxL3} and \eqref{structUxL4} that  $\cU_{x_0,L_k}$ is closed if the events $\mathcal T_{x_0,\hat L}, \mathcal Z_{x_0,\hat L}$, $\cM_{x_0,L}$ and $ \mathcal N_{x_0,L_-}$ are closed.

Since all the above sets rely on the property of a box being good, closedness of $\cU_{x_0,L_k}$ follows
from closedness of the event that a box $\Lambda_{L}$ is $(\omega,E,m,\zeta)$-good for any given fixed $E\in\RR$ and parameters $m$ and $\zeta$. The latter can be seen by an application of the resolvent identity. Alternatively,
we remark that, for any given $E\in\RR$, the inclusion map $i_{E}: \Omega \rightarrow \RR\times \Omega, \omega\mapsto
(E,\omega)$, is continuous. Therefore, closedness of the event that $\Lambda_{L}$ is $(\omega,E,m,\zeta)$-good for given $E$ follows from closedness of the set
\be\label{eom}
	\big\{ (E,\omega)\in\RR\times\Omega : \L_L \mbox{ is } (\omega,E,m,\zeta)\mbox{-good} \big\}
\ee
in $\RR\times\Omega$ with the product topology.
For this result, in turn, we refer to \cite[Remark~3.3]{GKuniv}.
\end{proof}


\section{Decay of generalised eigenfunctions}

This appendix is not needed to obtain the results in this paper. But it is instructive to recall
how the sets $\mathcal{U}_{x_{0},L}$ are used in \cite{GKuniv} to establish decay of generalised eigenfunctions.
In particular, this allows to shed light on the additional challenges for Bernoulli random variables and
explains the more involved structure of $\mathcal{U}_{x_{0},L}$ in this case.

We recall some notation from \cite[Section 5]{GKuniv} : fix $ {\nu} > \frac d 2$. Given $y \in \Rd$,
 $T_y$ denotes the operator on $\cH=\mathrm{L}^2(\mathbb{R}^d)$ given by multiplication by the function
$\RR^{d} \ni x \mapsto T_y(x):= \angles{ x-y}^{\nu}$. We call $\varphi\not=0$ a generalised eigenfunction of $H_{\om}$ with generalised
 eigenvalue ${E}$ if and only if $T^{-1}\varphi\in \Lp{\Rd}$ and
\be \label{eigb}
 \langle H_{\om} \varphi,\psi \rangle = {E} \langle \varphi,\psi \rangle
\quad  \mbox{for all} \quad\psi \in C_c^\infty(\Rd),
\ee
where $C_c^\infty(\Rd)$ denotes the space of arbitrarily often differentiable, compactly supported functions on $\Rd$. Given $E\in\RR$, $\Theta_\om(E)$ denotes the set of generalised eigenfunctions  of $H_\om$ associated to $E$.

 \begin{defn}\label{Ws}
 Given $\om\in \Omega$, $x \in \Rd$ and ${E} \in \RR$, we measure the concentration of generalised eigenfunctions associated to $E$ around the point $x$ by
 \begin{equation}\label{Wx}
W_{\om,x}({E}):=\begin{cases}
{\sup_{\psi \in \Theta_{\om}({E})} }
\ \frac {\| \Chi_{x}\psi \|}
{\|T_{x}^{-1}\psi \|}&
 \text{if $\Theta_{\om}({E})\not=\emptyset$}\\0 & \text{otherwise}
 \end{cases},
\end{equation}
and at an annulus around $x$ at scale $L\geq 1$ by
  \begin{equation} \label{WxL}
W_{\om,x,L}({E}):=\begin{cases} {\sup_{\psi \in \Theta_{\om}({E})} }
\frac {\norm{ \Chi_{x,L}\psi}}
{\norm{T_{x}^{-1}\psi }}&
 \text{if $\Theta_{\om}({E})\not=\emptyset$}\\0 & \text{otherwise}
.\end{cases},
\end{equation}
where $\Chi_{x,L}$ is the characteristic function of the annulus $\L_{2L+1}(x)\setminus\L_{L-1}(x)$.
\end{defn}

For simplicity, we state the following result in the particular case of H\"older continuous random variables, whose proof is given in \cite[Rem.\ 6.14]{GKuniv}:

\begin{thm}[Thm.\ 6.1 \cite{GKuniv}]\label{UxL}
 Let $H_\om$ be the Delone--Anderson Hamiltonian on $\Lp{\Rd}$ and assume the random variables are H\"older continuous. Given $p>3$ and $\zeta\in(0,1)$, consider a bounded open interval $\cI$ such that there is $m>0$ and a scale $\cL$ such that all scales  $L\geq \cL$ are $(E,m,\zeta,p)$-good for  all $E\in\cI$.
 Then, given a sufficiently large scale $L$ and any $x_0\in\Rd$, there exists an event $\cU_{x_0,L}$ with the following properties:
 \begin{itemize}
  \item[i)] We have
  \be\label{u1} \cU_{x_0,L}\in \cF_{\L_{L_+}(x_0)} \,\, \mbox{ and }\,\, \P\pa{\cU_{x_0,L}}\geq 1-L^{\frac{p-3}{3}d}. \ee
  \item[ii)] If $\om\in \cU_{x_0,L}$ and $E\in \cI$, whenever
  \be\label{w1} W_{\om,x_0}(E)> e^{-\frac{m}{30}\sqrt{L}},  \ee
  we have that
  \be\label{w2} W_{\om,x_0,L}(E) \leq e^{-\frac{m}{1000}\sqrt{L}}. \ee
  \item[iii)]  If $\om\in \cU_{x_0,L}$, we have
   \be\label{w3} W_{\om,x_0}(E) W_{\om,x_0,L}(E)\leq e^{-\frac{m}{1000}\sqrt{L}}\quad \mbox{for all } E\in\cI.\ee
 \end{itemize}
\end{thm}

 Theorem \ref{UxL} applies when $\frac{p-3}{3}>0$, which is the case for the Anderson model with H\"older continuous random variables, see \cite[Remark 1.7]{GKuniv}. In the case of Bernoulli random variables, however, the probability estimates for the multiscale analysis are weaker compared to the case of regular random variables.
Moreover, the use of unique continuation principles forces the restriction $p\in (\frac{1}{3},\frac{3}{8})$, see \cite[Remark 4.8]{GKuniv}. In this case a more involved proof is needed to obtain the analog of Theorem \ref{UxL}, which includes several scales to take care of the weaker estimates in order to reach conclusions (i), (ii) and (iii) above. Namely, in the case $0<p<3$ the analog of Theorem \ref{UxL} is given in \cite[Thm.\ 6.1]{GKuniv}, where the interval $\mathcal I$ is replaced by a slightly smaller interval $\mathcal I_L$, and the r.h.s. of Equations \eqref{u1},\eqref{w1},\eqref{w2} is replaced by
\be \P\pa{\cU_{x_0,L}}\geq 1-L^{\tilde p d},\quad \tilde p\in (0,p), \ee
\be  W_{\om,x_0}(E)> e^{-M{L^\nu}}\quad \mbox{for an appropriate }\nu\in(0,1), \mbox{and }M>0,   \ee
 \be W_{\om,x_0,L}(E) \leq e^{-M{L}}. \ee
 Consequently,  $W_{\om,x_0}(E) W_{\om,x_0,L}(E)\leq e^{-\frac{M}{2}{L^\nu}}$ for all energies $E$ in a set slightly smaller than $\cI$.
The extra steps needed to treat the Bernoulli case explain the more involved structure of the event $\cU_{x_0,L}$ in  \eqref{structUxL2}.

\end{appendix}

%
\section*{Acknowledgements}

The authors thank I.\ Veseli\'c for pointing out the existence of a one-dimensional
quantitative unique continuation principle in references \cite{KV02,V96}, which leads to Theorem \ref{t:ucp1}.
The authors are grateful to F.\ Germinet for many enjoyable discussions, helpful remarks and for his support
throughout the development of this work.
CRM received financial support from the Hausdorff Center for Mathematics (Bonn) and CRC 1060.

%


\begin{thebibliography}{BdMNSS}
\frenchspacing

\bibitem[AW]{AW} Aizenman, M., Warzel, S.:
	\emph{Random operators: disorder effects on quantum spectra and dynamics},
	Graduate Studies in Mathematics 168, American Mathematical Society, Providence, RI, 2015.

\bibitem[B]{Bel03} Bellissard, J.: Noncommutative geometry of aperiodic solids. In
	Cardona, A., Paycha, S., Ocampo, H. (Eds.):
	\emph{Geometric and topological methods for quantum field theory},
	World Scientific Publishing, River Edge, NJ, 2003, pp.\ 86--156.

\bibitem[BIST]{BIST89} Bellissard, J., Iochum, B., Scoppola, E., Testard, D.:
	Spectral properties of one-dimensional quasi-crystals,
	\emph{Commun. Math. Phys.} 125, 527--543 (1989).



\bibitem[BoK]{BK} Bourgain, J., Kenig, C.:
	On localization in the continuous Anderson--Bernoulli model in higher dimension,
	\emph{Invent. Math.} 161, 389--426 (2005).

\bibitem[BouLS]{BdMLS11} Boutet de Monvel, A., Lenz, D., Stollmann, P.:
	An uncertainty principle, {W}egner estimates and localization near fluctuation boundaries,
    \emph{Math. Z.} 269, 663--670 (2011).

\bibitem[BouNSS]{BdMNSS06} Boutet de Monvel, A., Naboko, S., Stollmann, P., Stolz, G.:
	Localization near fluctuation boundaries via fractional moments and applications,
    \emph{J. Anal. Math.} 100, 83--116 (2006).
	
\bibitem[CKM]{CKM} Carmona, R., Klein, A., Martinelli, F.:
 	Anderson localization for Bernoulli and other singular potentials,
 \emph{Commun. Math. Phys.} 108, 41--66 (1987).

\bibitem[CL]{CaLa90}  Carmona, R, Lacroix, J.:
 	\emph{Spectral theory of random Schr\"odinger operators},
	Birkh\"auser, Boston, 1990.

\bibitem[CoH]{CH1} Combes,  J.M.,   Hislop, P.D.:
	Localization for some continuous, random Hamiltonians in $d$-dimensions,
	\emph{J. Funct. Anal.} {124}, 149--180 (1994).

\bibitem[DSS]{DSS02} Damanik, D., Sims, R., Stolz, G.:
	Localization for one-dimensional, continuum, Bernoulli--Anderson models,
	\emph{Duke Math. J.} 114, 59--100 (2002).

\bibitem[DaKW]{DKW17} Davey, B., Kenig, C., Wang, J.-N.:
	The Landis conjecture for variable coefficient second-order elliptic PDEs,
	\emph{Trans. Amer. Math. Soc.} 369, 8209--8237 (2017).

\bibitem[vDrK]{vDK89} von Dreifus, H., Klein, A.:
	A new proof of localization in the Anderson tight binding  model,
	\emph{Commun.  Math.  Phys.} 124, 285--299 (1989).

\bibitem[EK]{EK} Elgart, A., Klein, A.:
	Ground state energy of trimmed discrete Schr\"odinger operators and localization for trimmed Anderson models,
	\emph{J. Spectr. Theory}  4, 391--413 (2014).

\bibitem[ES]{ES} Elgart, A., Sodin, S.:
	The trimmed Anderson model at strong disorder: localization and its breakup,
	\emph{J. Spectr. Theory} 7, 87--110 (2017).

\bibitem[FS]{FS}  Fr\"ohlich, J., Spencer, T.:
	Absence of diffusion in the Anderson tight binding model for large disorder or low energy,
	\emph{Commun. Math. Phys.} 88, 151--184 (1983).

\bibitem[FuK]{FuKr88} Fujiwara, T., Kraj\v{c}{\'i}, M.:
	Strictly localized eigenstates on a three-dimensional Penrose lattice,
   \emph{Phys. Rev. B} 38, 12903--12907 (1988).

\bibitem[G]{G} Germinet, F.:
	Recent advances about localization in continuum random Schr\"odinger operators with an extension
	to underlying Delone sets. In
	Beltita, I., Nenciu, G., Purice, R. (Eds.):
	\emph{Mathematical results in quantum mechanics},
	World Scientific Publishing, Hackensack, NJ, 2008, pp. 79--96.

\bibitem[GK1]{GK1}  Germinet, F., Klein, A.:
	Bootstrap multiscale analysis and localization in random media,
	\emph{Commun. Math. Phys.} {222}, 415--448 (2001).

\bibitem[GK2]{GK2} Germinet, F., Klein, A.:
	Explicit finite volume criteria for localization in continuous random media and applications,
	\emph{Geom. Funct. Anal.} 13, 1201--1238 (2003).

\bibitem[GK3]{GKpams} Germinet, F., Klein, A.:
	Operator kernel estimates for functions of generalized Schr\"odinger operators,
	\emph{Proc. Amer. Math. Soc.} 131, 911--920 (2003).


\bibitem[GK4]{GKuniv} Germinet, F., Klein, A.:
	A comprehensive proof of localization for continuous {Anderson} models with singular random potentials,
	\emph{J. Eur. Math. Soc. (JEMS)} 15, 55--143 (2013).

\bibitem[GMRM]{GMRM} Germinet, F., M{\"u}ller, P., Rojas-Molina, C.:
	Ergodicity and dynamical localization for Delone--Anderson operators,
	\emph{Rev. Math. Phys.} 27, 1550020, 1--36   (2015).

\bibitem[GiT]{GT} Gilbarg, D., Trudinger, N. S.:
	\emph{Elliptic  partial  differential  equations  of  second  order},
	Springer, Berlin, 1983.

\bibitem[GrHV]{GHV08} Gruber, M.J., Helm, M.,  Veseli\'c, I.:
	Optimal Wegner estimates for random Schr\"odinger operators on metric graphs.
	In Exner, P., Keating, J.P., Kuchment, P., Sunada, T., Teplyaev, A. (Eds.):
	\emph{Analysis on graphs and its applications}, Proc. Symp. Pure Math. 77,
	American Mathematical Society, Providence, RI, 2008, pp. 409--422.

\bibitem[HV]{HV07} Helm, M., Veseli\'c, I.:
	Linear Wegner estimate for alloy-type Schr\"odinger operators on metric graphs,
	\emph{J. Math. Phys.} 48, 092107, 1--7 (2007).

\bibitem[KH]{KaHa} Katok, A., Hasselblatt, B.:
	\emph{Introduction to the modern theory of dynamical systems},
	Cambridge University Press, Cambridge, 1995.

\bibitem[Ki]{K} Kirsch, W.:
	An invitation to random Schr\"odinger operators. In
	\emph{Random Schr\"odinger Operators},
	Panor. Synth\`eses, vol. 25, Soc. Math. France, Paris, 2008, pp. 1–119.

\bibitem[KiK]{KK} Kirsch, W., Krishna, M.:
	Spectral statistics for Anderson models with sporadic potentials,
	\emph{J. Spectr. Theory } 10, 581--597 (2020).

\bibitem[KiV]{KV02} Kirsch, W., Veseli\'c, I.:
	Existence of the density of states for one-dimensional alloy-type potentials with small support,
	\emph{Contemp. Math.} 307, 171--176 (2002).

\bibitem[KlLS1]{KLS03a} Klassert, S., Lenz, D., Stollmann, P.:
	Delone dynamical systems: ergodic features and applications. In
	Trebin, H.-R. (Ed.):
	\emph{Quasicrystals},
	Wiley-VCH, Weinheim, 2003, pp. 172--187.

\bibitem[KlLS2]{KLS03b} Klassert, S., Lenz, D., Stollmann, P.:
	Discontinuities of the integrated density of states for random operators on Delone sets,
	\emph{Commun. Math. Phys.} 241, 235--243 (2003).

\bibitem[KlLS3]{KLS11} Klassert, S., Lenz, D., Stollmann, P.:
	Delone measures of finite local complexity and applications to spectral theory of one-dimensional
	continuum models of quasicrystals,
	\emph{Discrete Contin. Dyn. Syst.} 4, 1553--1571 (2011).

\bibitem[KleLSp]{KLS90} Klein, A., Lacroix, L., Speis, A.:
	Localization for the Anderson model on a strip with singular potentials,
  \emph{J. Funct. Anal.} 94, 135--155 (1990).

\bibitem[KloLNS]{KLNS10} Klopp, F., Loss, M., Nakamura, S., Stolz, G.:
	Localization for the random displacement model,
	\emph{Duke Math. J.} 161, 587--621 (2012).
	
\bibitem[KoS]{KoSu86} Kohmoto, M., Sutherland, B.:
	Electronic states on a Penrose lattice,
	\emph{Phys. Rev. Lett.} 56, 2740--2743 (1986).

\bibitem[KuV]{KuV} Kuchment, P., Vainberg, B.:
	On absence of embedded eigenvalues for Schrö\-dinger operators with perturbed periodic potentials,
 	\emph{Commun. Partial Differential Equations} 25, 1809--1826 (2000).

\bibitem[LS1]{LS02} Lenz, D., Stollmann, P.:
	Quasicrystals, aperiodic order and groupoid von {Neumann} algebras,
	\emph{C. R. Acad. Sci. Paris, Ser. I} 334, 1131--1136 (2002).

\bibitem[LS2]{LS03} Lenz, D., Stollmann, P.:
	Delone dynamical systems and associated random operators. In
	Combes, J.-M., Cuntz, J., Elliott, G.A., Nenciu, G., Siedentop, H., Str\v{a}til\v{a}, S. (Eds.):
  \emph{Operator algebras and mathematical physics},
  Theta, Bucharest, 2003, pp. 267--285.

\bibitem[LS3]{LS05} Lenz, D., Stollmann, P.:
	An ergodic theorem for Delone dynamical systems and existence of the integrated density of states,
  \emph{J. Anal. Math.} 97, 1--24 (2005).

\bibitem[LS4]{LS06} Lenz, D., Stollmann, P.:
 	Generic sets in spaces of measures and generic singular continuous spectrum for Delone Hamiltonians,
  \emph{Duke Math. J.} 131, 203--217 (2006).

\bibitem[LV]{LV09} Lenz, D., Veseli\'c, I.:
	Hamiltonians on discrete structures: jumps of the integrated density of states and uniform convergence,
	\emph{Math. Z.} 263, 813--835 (2009).

\bibitem[Li]{Liu} Liu, W.:
	Irreducibility of the Fermi variety for discrete periodic Schr\"odinger operators and embedded eigenvalues,
	preprint \href{https://arxiv.org/abs/2006.04733}{arXiv:2006.04733} (2020).

\bibitem[LiO]{LiuOng} Liu, W., Ong, D. C.:
	Sharp spectral transition for eigenvalues embedded into the spectral bands of perturbed periodic operators,
	\emph{J. Anal. Math.}, 141, 625--661 (2020).

\bibitem[MR]{MR} M\"uller, P., Richard, C.:
	Ergodic properties of randomly coloured point sets,
	\emph{Canad. J. Math.} 65, 349--402 (2013).


\bibitem[PF]{PF92} {Pastur, L.}, Figotin, A.:
	\emph{Spectra of random and almost-periodic operators},
  Springer, Berlin, 1992.

\bibitem[RS]{RSI} Reed, M., Simon, B.:
	\emph{Methods of modern mathematical physics III: scattering theory},
	Academic Press, New York, 1979.

\bibitem[RoM1]{RMthesis} Rojas-Molina, C.:
	\emph{Etude math\'ematique des propri\'et\'es de transport des op\'era\-teurs de Schr\"odinger al\'eatoires
		avec structure quasi-cristalline},
	PhD thesis, Universit\'e de Cergy-Pontoise, 2012.
	Available at \href{http://www.crojasmolina.com}{www.crojasmolina.com}.

\bibitem[RoM2]{RM12} Rojas-Molina, C.:
	Characterization of the Anderson metal-insulator transition for non ergodic operators and application,
  \emph{Ann. Henri Poincar\'e} 13, 1575--1611 (2012).

\bibitem[RoM3]{RM13} Rojas-Molina, C.:
	The Anderson model with missing sites,
	\emph{Oper. Matrices} 8, 287--299 (2014).

\bibitem[RoM4]{RM20} Rojas-Molina, C.:
	Random Schr\"odinger operators and Anderson localization in aperiodic media,
	\emph{Rev. Math. Phys.} 33,  2060010, 1--11 (2021). 

\bibitem[RoMV]{RMV12} Rojas-Molina, C., Veseli\'c, I.:
	Scale-free unique continuation estimates and applications to random Schr\"odinger operators,
	\emph{Commun. Math. Phys.} 320, 245--274 (2013).

\bibitem[ST]{SeelmanTaufer19} Seelmann, A., T\"aufer, M.:
	Band edge localization beyond regular Floquet eigenvalues,
	\emph{Ann. Henri Poincar\'e} 21, 2151--2166 (2020).

\bibitem[ShBGC]{Sch84} Shechtman, D., Blech, I., Gratias, D., Cahn, J. W.:
	Metallic phase with long-range orientational order and no translation symmetry,
	\emph{Phys. Rev. Lett.} 53, 1951--1953 (1984).

\bibitem[Si]{Si} Simon, B.:
	Operators with singular continuous spectrum: I. General operators,
	\emph{Ann. Math. (2)} 141, 131--145 (1995).

\bibitem[St]{S} Stollmann, P.:
	\emph{Caught by disorder: bound states in random media},
	Birkh{\"a}user, Boston, 2001.

\bibitem[S\"u]{Su89} S{\"u}to, A.:
	Singular continuous spectrum on a Cantor set of zero Lebesgue measure for the Fibonacci Hamiltonian,
	\emph{J. Stat. Phys.} 56, 525--531 (1989).

\bibitem[V]{V96} Veseli\'c, I.:
	\emph{Lokalisierung bei zuf\"allig gest\"orten periodischen Schr\"odinger\-opera\-toren in Dimension Eins}, Diplomarbeit, Mathematisches Institut der Ruhr-Universit\"at, Bochum, 1996.



\end{thebibliography}
\end{document}